\documentclass[journal,twocolumn,10pt]{IEEEtranTCOM}
\normalsize
\usepackage{graphicx} 
\usepackage[cmex10]{amsmath}
\usepackage{amssymb}
\usepackage{amsfonts}
\usepackage{caption}
\usepackage{subcaption}
\usepackage{amsthm}
\usepackage{psfrag}
\usepackage{array}
\usepackage{epstopdf}
\usepackage{multirow}
\usepackage{multicol}
\usepackage{float}
\usepackage{tablefootnote}
\usepackage{threeparttable,booktabs}
\usepackage{mathtools}
\usepackage{mdwmath}
\usepackage{mdwtab}
\usepackage{algorithm2e}
\usepackage{filecontents}
\usepackage{tablefootnote}
\newtheorem{theorem}{Theorem}
\newtheorem{lemma} {Lemma}

\newtheorem{corollary}{Corollary}
\usepackage{lipsum}
\usepackage{mathtools}
\usepackage{footnote} 
\makesavenoteenv{tabular} 
\usepackage{cuted}

\newcommand{\bm}[1]{\mathbf{#1}}
\newcommand{\bsy}[1]{\boldsymbol{#1}}

\newcommand{\be}{\begin{equation}}
\newcommand{\ee}{\end{equation}}
\newcommand{\bea}{\begin{eqnarray}}
\newcommand{\eea}{\end{eqnarray}}

\newcommand{\z}{{\bm z}}

\newcommand{\dd}{{\bm d}}
\newcommand{\bA}{{\bm A}}
\newcommand{\bI}{{\bm I}}

\newcommand{\bD}{{\bf D}}

\newcommand{\bS}{{\bf S}}
\newcommand{\bH}{{\bf H}}

\newcommand{\bR}{{\bf R}}

\newcommand{\bg}{{\bf g}}
\newcommand{\by}{{\bf y}}

\newcommand{\bh}{{\bf h}}

\newcommand{\bU}{{\bf U}}

\newcommand{\bd}{{\bf d}}

\newcommand{\bx}{{\bf x}}

\newcommand{\bnu}{\mbox{\boldmath$\nu$}}

\newcommand{\AHA}{\mbox{$\bA^{\rm H}\bA$}}

\newcommand{\mI}{\mbox{$\frac{\sigma_\nu^2}{\sigma_d^2}\bI$}}
\newcommand{\snr}{\mbox{$\frac{\sigma_d^2}{\sigma_\nu^2}$}}
\newcommand{\snri}{\mbox{$\frac{\sigma_\nu^2}{\sigma_d^2}$}}

\newcommand{\modulation}{\mathcal{M}}
\newcommand{\timeshift}{\mathcal{T}}
\newcommand{\cexp}[1]{e^{\frac{j2\pi {#1}}{N}}}
\newcommand{\cexpM}[1]{e^{\frac{j2\pi {#1}}{M}}}
\newcommand{\cexpMN}[1]{e^{\frac{j2\pi {#1}}{MN}}}
\newcommand{\bdft}{\bm F_b}
\newcommand{\tempvec}{\bsy \vartheta}
\newcommand{\dftmat}{\bm W}
\newcommand{\permut}{\bm P}
\newcommand{\identity}{\bm I}
\newcommand{\kron}{\bsy \otimes}
\newcommand{\diagdft}[1]{ \bsy U_{#1}}
\newcommand{\matrixmmsebias}[1]{\bsy \Theta_{#1}}
\newcommand{\channelfreqmatrix}{\bsy \Lambda}
\newcommand{\channelfreqcoef}[1]{\tilde{h}(#1)}
\DeclarePairedDelimiter\floor{\lfloor}{\rfloor}

\newcommand{\diagonalgfdm}{\bar{\bm D}}
\newcommand{\unitarymatrix}{\bsy \Gamma}

\newcommand{\flopdft}[1]{\text{FL}_{#1}}

\ifCLASSINFOpdf
\else
\fi

\usepackage{fixltx2e}

\usepackage{stfloats}
\begin{document}
\bstctlcite{IEEEexample:BSTcontrol}
%
\title{Design of Low Complexity GFDM Transceiver}

 \author{\IEEEauthorblockN{Shashank Tiwari and Suvra Sekhar Das\\}
 \IEEEauthorblockA{
Indian Institute of Technology Kharagpur, India.
}}

\maketitle

\begin{abstract}

 In this work, we propose a novel low complexity Generalised Frequency Division Multiplexing (GFDM) transceiver design. GFDM modulation matrix is factorized into FFT matrices and a diagonal matrix to design low complexity GFDM transmitter. Factorization of GFDM modulation matrix is used to derive low complexity Matched Filter (MF), Zero Forcing (ZF) and Minimum Mean Square Error (MMSE) based novel low complexity self-interference equalizers. A two-stage receiver is proposed for multipath fading channel in which channel equalization is followed by our proposed low-complexity self-interference equalizers. Unlike other known low complexity GFDM transceivers, our proposed transceiver attains low complexity for arbitrary number of time and frequency slots.  The complexity of our proposed transceiver is log-linear with a number of transmitted symbols and achieves 3 to 300 times lower complexity compared to the existing structures without incurring any performance loss. Our proposed Unbiased-MMSE receiver outperforms our proposed ZF receiver without any significant increase in complexity especially in the case of large number of time slots.  In a nutshell, our proposed transceiver enables low complexity flexible GFDM transceiver implementation.
\end{abstract}

 \ifCLASSOPTIONpeerreview
 \begin{center} \bfseries EDICS Category: 3-BBND \end{center}
\fi
%
\IEEEpeerreviewmaketitle
\section{Introduction}

GFDM is a  block based waveform fits into many next generation cellular network requirements \cite{banelli_modulation_2014,andrews_what_2014} such as low OoB radiation \cite{michailow_generalized_2014}, immunity to CFO \cite{matthe_asynchronous_2015,choi_effect_2015}, compatibility with multiple input multiple output (MIMO) \cite{matthe_near-ml_2015,yenilmez_performance_2016,matthe_multi-user_2015} and  flexibilty to use different time-frequency slots and pulse shapes \cite{mendes_gfdm:_2016}. However, it requires transceivers with high computational complexity. 
 This is due to non-orthogonality, which is introduced by the circular filtering of each sub-carrier. Moreover, GFDM suffers from self-interference which mandates the use of complex receivers to equalize self-interference. When exposed to multipath channel, GFDM signal further distorts which increases the complexity of signal reconstruction.  
 
A two-stage receiver can be used for GFDM reception in multipath fading channel in which channel equalization is followed by  self-interference equalizers \cite{michailow_bit_2012}. Channel equalization can be implemented by using well known low complexity frequency domain equalizaers (FDE)  same as in the case of orthogonal frequency devision multiplexing (OFDM). However, implementation of self-intereference equalization is costly \cite{farhang_low-complexity_2016}. If $M$ and $N$ represent number of time and frequency slots respectively, the implementation of the transmitter, Matched Filter (MF) self-interference equalization and Zero-Forcing (ZF) self-interference equalization involves a complexity of $O(M^2N^2)$ \cite{tiwari_precoded_2015} while the complexity of  Minimum Mean Square Error (MMSE) self-interference equalization is $O(M^3N^3)$. When $N\sim10^3$'s and $M\sim10$'s (or $N\sim 10$'s and $M\sim 10^3$'s), the count of computations becomes very high. This high complexity hinders  practical implementation of  GFDM transceivers. Therefore if complexity of GFDM transceivers can be reduced it would help widespread use of this versatile waveform design framework     \cite{mendes_gfdm:_2016}. 
  
  The sparsity of  prototype pulse shape in frequency domain is exploited to design a low complexity transmitter in \cite{michailow_generalized_2012} and a low-complexity MF receiver in \cite{gaspar_low_2013}. The complexity is reduced to $O(MNlog_2(MN)+MN^2)$ but it comes with increase in BER. Periodicity of complex exponential is exploited in \cite{lin_orthogonality_2015,matthe_precoded_2016} to reduce the complexity further to $O(MN\log_2(N)+M^2N)$. Similar order of complexity is achived by using  block circulant property of multiplication of modulation matrix and its Hermitian  in \cite{farhang_low-complexity_2016}. Behrouz and Hussein proposed frequency spreading based GFDM transmitter in \cite{farhang-boroujeny_derivation_2015} based on the principles of frequency spreading filter bank multi carrier (FMBC) transmitter proposed in \cite{bellanger_physical_2010}. The complexity of the transmitter is $O(MN\log_2(N)+M^2N)$ similar to \cite{farhang_low_2015, matthe_precoded_2016,lin_orthogonality_2015}. Recently, Wei et al. in \cite{wei_fast_2016} have proposed a low complexity one-stage receiver based on frequency-domain discrete Gabor transform (FD-DGT) called Localized DGT receiver (LDGT) having the complexity of $O(MN\log(MN))$. Authors in \cite{wild_reduced_2015} use two assumptions for designing low complexity receiver (i) requirement of perfect knowledge of coherence bandwidth by LDGT receiver  \footnote{See equation (50-52) in \cite{wei_fast_2016}. Further, while computing complexity of LDGT, the dual function $\tilde{\bm \Gamma}_{opt}$ in (50) is assumed to be known for a broadband channel.} and (ii) subcarrier bandwidth to be less than channel coherence bandwidth\footnote{See equation (15) in \cite{wei_fast_2016}}. Performance of LDGT receiver will depend on estimation of coherence bandwidth which is further dependent upon stationarity conditions \cite{siamarou_broadband_2003, wang_propagation_2013}. There will be additional complexity requirement to estimate coherence bandwidth. Constraint on sub-carrier bandwidth decreases the flexibility of GFDM.
  
   A flexible GFDM system is free to choose any arbitrary value of $M$ and $N$, arbitrary pulse shape and arbitrary subcarrier bandwidth \cite{michailow_generalized_2014} for different application requirements \cite{mendes_gfdm:_2016}. To the best of author's knowledge complexity of  GFDM transmitter and linear receiver is either found to increase non-linearly with $M$ or increase non-linearly with $N$. Additionally, only transceiver present in \cite{farhang_low-complexity_2016} assumes arbitrary pulse shape but has high complexity when $M$ is high.  Hence, no existing GFDM transceiver enables flexibility of GFDM. Hence, to enable the flexibility of GFDM,  we aim to reduce the complexity growth further on $M$ and $N$. 
  
  \subsection*{Our Contribution}
 In this work, we assume prior knowledge of GFDM specific parameters such as $N$, $M$ and pulse shape at transmitter as well as receiver like in \cite{michailow_generalized_2012,gaspar_low_2013,lin_orthogonality_2015,matthe_precoded_2016,farhang_low-complexity_2016,wei_fast_2016}. Our design consider arbitrary value $N$ and $M$. Unlike \cite{gaspar_low_2013, michailow_generalized_2014}, our design assumes arbitrary pulse shapes. We consider that subcarrier bandwidth can be smaller as well as larger than channel coherence bandwidth, unlike in \cite{wei_fast_2016}.
  \begin{enumerate}
  \item We factorize the modulation matrix of GFDM in terms of a Block Circulant Matrix and a Block Diagonal IFFT matrix. This lead us to FFT based low complexity GFDM transmitter implementation.
  \item We derive closed form expression for MF, ZF and MMSE self-interference equalizers using above-mentioned GFDM modulation matrix factorization. These closed form expressions provide low complexity FFT based low complexity implementations.
  \item We also derive closed form expression for bias correction correction for MMSE self-equalizer output to improve BER performance of biased-MMSE receiver\cite{cioffi_mmse_1995}.
  \item We present low complexity transceiver structure in multipath channel. The overall complexity of our transceiver is found to be $O(MN\log_2(MN))$ which is significantly below that of existing transceivers. 
  \item We compare BER performance of our proposed transceiver with the direct implementation of receivers. Our proposed low complexity tranceiver does not make any assumption related to parameters of GFDM i.e. number of time-frequency slots and pulse shape. Hence, with our transciever the properties of GFDM is same as in \cite{michailow_generalized_2014}. Thus, our trancievers attain same BER performance as direct implementation of GFDM trancievers in \cite{michailow_generalized_2014}.
  \end{enumerate}
  

\section{System Model}\label{sec:systemmodel}
\begin{figure*}[t]
\includegraphics[width=\linewidth]{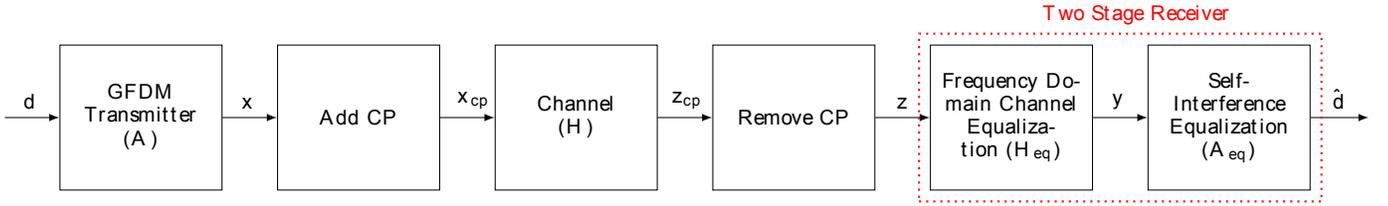}
\caption{Baseband Block Diagram of GFDM Transceiver in Multipath Fading Channel}
\label{fig:systemmodel}
\end{figure*}
In this work, vectors are represented by bold small letters ($\bx$), matrices are represented by bold capital letters ($\bm X$) and scalers are represented as normal small letters ($x$). $j=\sqrt{-1}$.  
The superscripts $(.)^{\rm T}$and $(.)^{\rm H}$ indicate transpose and conjugate transpose  operations, respectively. Table~\ref{tab:listofsymbols} lists operators and important symbols used in rest of the paper.
\begin{table}[h]
\centering
\caption{List of Important Symbols and Operators} \label{tab:listofsymbols}
\begin{tabular}{| >{\centering\arraybackslash}m{0.25\linewidth}| >{\centering\arraybackslash}m{0.6\linewidth}|} \hline
\textbf{Operators} & \textbf{Description} \\ \hline \hline
$\kron$ & Kronecker product operator  \\ \hline
$diag\{.\}$  &
A diagonal matrix whose diagonal elements
are formed by the elements of the vector inside or diagonal elements of the matrix inside \\ \hline
$\mod$ & Modulus Operator \\ \hline
$\floor*{.}$ & Rounds the value inside to the nearest integer towards minus infinity \\ \hline \hline
\textbf{Symbols} & \textbf{ Description} \\ \hline \hline
$\bI_{\{.\}}$ & Identity matrix with order $\{.\}$ \\ \hline
$\dftmat_{\{.\}}$  & $\{.\}$-order normalized IDFT matrix \\ \hline
$M$ & Number of Time Slots \\ \hline
$N$ & Number of sub-carriers \\ \hline
$g[n]$, $n=0,~1, \cdots MN-1$. & prototype filter coefficients \\ \hline
$\bd$ & $MN$ length data vector \\ \hline
$\bA$ & GFDM Modulation Matrix \\ \hline
$\bx$ & transmitted GFDM vector \\ \hline
$\bx_{cp}$ & transmitted GFDM vector after CP addition \\ \hline
$\bh$ & channel impulse response vector \\ \hline
$\bH$ & circulant channel convolution matrix \\ \hline 
$\z_{cp}$ & received vector \\ \hline
$\z$ & received vector after CP removal \\ \hline
$\channelfreqcoef{r}$, $r=0,~1,\cdots,~ MN-1$ & channel frequency coefficients \\ \hline
$\channelfreqmatrix$ & Channel frequency diagonal Matrix \\ \hline
$\bnu$ & AWGN noise vector \\ \hline
$\channelfreqmatrix_{eq}$ & Channel equalization diagonal matrix for FDE \\ \hline
$\hat{\dd}$ & Estimated data vector \\ \hline 
$\bA_{eq}$ & Self-Interference Equalization Matrix \\ \hline
$\matrixmmsebias{gfdm}$ & Bias correction diagonal Matrix for MMSE self-interference equalization\\ \hline
$\bm{G}$ & Block circulant Matrix with diagonal matrix of order $N$ \\ \hline
$\bU_N$ & Block diagonal matrix with IDFT blocks of order $N$ \\ \hline
$\bD$ & $MN$-order diagonal matrix \\ \hline
$\permut$ & Permutation Matrix \\ \hline   
\end{tabular}
\end{table}

\subsection{Transmitter}
We have a GFDM system with $N$ sub-carriers and $M$ timeslots. The $MN$ length prototype filter is $g(n), n=0,~1,\ldots,MN-1$. QAM modulated data symbol is $d_{m,k}\in \mathbb{C},~m=0,1,\ldots, M-1, ~k=0,1,~\ldots,N-1$. We assume that data symbols are independent and identical i.e. $E[d_{m,k} d_{m',k'}^\ast] = \sigma_d^2 \delta_{m-m',k-k'} $. The transmitted GFDM signal can be written as,
\begin{equation}
x[n] = \frac{1}{\sqrt{N}}\sum_{m=0}^{M-1}{\sum_{k=0}^{N-1}{d_{m,k}g[n-mN]_{MN}e^{\frac{j2\pi n k}{N}}}}.
\end{equation}
The transmitted signal can also be written as \cite{michailow_generalized_2014},
\begin{equation}
\label{eqn:transmittedsignalmatrixform}
\bx= \bA_{MN\times MN} \bd_{MN \times 1},
\end{equation}
where $\bd=[\bd_0 ~ \bd_1 \ldots \bd_{M-1}]^{\rm T}$ is the data vector, where $\bd_m=[d_{m,0}~d_{m,1}\ldots d_{m,N-1}]^{\rm T}$, where, $m=0,~1\ldots M-1$, is the $N$ length data vector for $m^{\rm th}$ time slot and $\bA$ is the modulation matrix which can be given as,
\begin{equation} \label{eq:Amat}
\begin{aligned}
\bA &= [ \bg ~ ~\modulation_1 \bg  ~  \cdots ~ \modulation_{N-1} \bg ~ |  \timeshift_1 \bg~ ~ \timeshift_1 \modulation_1 \bg ~ \cdots  \timeshift_1 \modulation_{N-1} \bg | \\
&\cdots |  \timeshift_{M-1} \modulation_1 \bg  \cdots \timeshift_{M-1} \modulation_{N-1} \bg ],
\end{aligned}
\end{equation}
where, $\bg=[g[0]~g[1]~\cdots~g[MN-1]]^{\rm T}$ is $MN$ length vector which holds the prototype filter coefficients, $\modulation_l g[n]=g[n]\cexp{ln}$ is the modulation operator and $\timeshift_r=\bg(n-rN)_{MN}$ is the cyclic shift operator.

  CP of length $N_{CP}$ is prepended to $\bx$. After adding CP, transmitted vector, $\bx_{cp}$, can be given as,
\be \label{x_cp}
\bx_{cp} = [\bx(MN-N_{cp}+1 : MN)\hspace{0.2 cm} ;~ \bx].
\ee
\subsection{Channel}
  Let, $\bh=[h_0,~h_1,\cdots h_{L-1}]^{\rm T}$ be $L$ length  channel impulse response vector, where, $h_i,~\text{for}~0 \leq i \leq L-1$, represents the complex baseband channel coefficient of $(i+1)^{\rm th}$ path \cite{sklar_digital_2001}, which we assume is zero mean circular symmetric complex Gaussian (ZMCSC). We also assume that channel coefficients related to different paths are uncorrelated. We consider, $N_{cp}\geq L$. Received vector of length $N_{CP}+NM+L-1$ is given by,
\be \label{y_cp}
\z_{cp} = \bh \ast \bx_{cp} + \bnu_{cp},
\ee
where $\bnu_{cp}$ is AWGN vector of length $MN+N_{cp}+L-1$ with elemental variance $\sigma_\nu^2$. 
\subsection{Receiver}
 The first $N_{cp}$ samples and last $L-1$ samples of $\by_{cp}$ are removed at the receiver i.e. $\by = [\by_{cp}(N_{cp}+1 : N_{cp}+MN)]$. Use of cyclic prefix converts linear channel convolution to circular channel convolution when $N_{cp}\geq L$\cite{prasad_ofdm_2004}. The $MN$ length received vector after removal of CP can be written as,
  \be \label{eqn: rec}
  \begin{aligned}
  \z &= \bH \bA \dd + \bnu,
  \end{aligned}
  \ee
where $\bH$ is  circulant convolution matrix of size $MN \times MN$ and $\bnu$ is WGN vector of length $MN$ with elemental variance $\sigma_\nu^2$. Since $\bH$ is a circulant matrix, $\by$ can be further written as,
\be \label{eqn: rec}
\z = \dftmat_{MN} \channelfreqmatrix \dftmat_{MN}^{\rm H}\bA \dd + \bnu,
\ee
where,  $\channelfreqmatrix=diag\{\channelfreqcoef{0}~,\channelfreqcoef{1}\cdots \channelfreqcoef{MN-1}\}$  is a diagonal channel frequency coefficients matrix whose $r^{\rm th}$ coefficient can be given as,  $\channelfreqcoef{r}=\sum_{s=0}^{L-1}{h(s)\cexpMN {sr}}$ where, $r=0,~1\cdots MN-1$.

In this work, we consider two stage receiver in which channel equalization is followed by  GFDM demodulation \cite{michailow_generalized_2014, gaspar_low_2013,farhang_low-complexity_2016,michailow_bit_2012}. Channel equalized vector, $\by$, can be given as \cite{sari_frequency-domain_1994},

\be \label{eqn:systemmodel:channelequalized}
\begin{aligned}
\by &= \dftmat_{MN} \channelfreqmatrix_{eq} \dftmat_{MN}^{\rm H} \z &= a\bA \dd + \bsy{b}+\tilde{\bnu},
\end{aligned} 
\ee
where, $\channelfreqmatrix_{eq}=
\begin{cases}
\channelfreqmatrix^{-1} ~ \text{for ZF FDE}\\
[\channelfreqmatrix^{\rm H}\channelfreqmatrix+\snri \bI_{MN}]^{-1} \channelfreqmatrix^{\rm H} ~ \text{for MMSE FDE}
\end{cases}$
where, $\tilde{\bnu}= \dftmat_{MN} \channelfreqmatrix_{eq} \dftmat_{MN}^{\rm H} \bnu$, 
$$ a=\begin{cases} 
1 ~ ~~~ \text{for ZF FDE} \\
\frac{1}{MN}\sum_{r=0}^{MN-1}{{\frac{|\channelfreqcoef{r}|^2}{|\channelfreqcoef{r}|^2+\snri}}} ~ \text{for MMSE FDE,} \end{cases}$$
$\bsy{b}$ is residual interference, given in (\ref{eqn:b}) and $\tilde{\bnu}=\dftmat_{MN} \channelfreqmatrix_{eq} \dftmat_{MN}^{\rm H}\bnu $ is post-processing noise. 

\begin{equation} \label{eqn:b}
\bsy{b}= \begin{cases} 
\bsy{0} ~ ~~~ \text{for ZF FDE} \\
[\dftmat_{MN} \channelfreqmatrix_{eq} \dftmat_{MN}^{\rm H}-\frac{1}{MN}\sum_{r=0}^{MN-1}{{\frac{|\channelfreqcoef{r}|^2}{|\channelfreqcoef{r}|^2+\snri}}}\bI_{MN}] \bA \dd \\ ~  ~~~~~~~~~~\text{for MMSE FDE} \end{cases}
\end{equation}


Channel equalized vector, $\by$, is further equalized to remove the effect of self-interference. Estimated data, $\hat{\bd}$, can be given as,
\be \label{eqn:gfdmequalizedsignal}
\hat{\bd} = \bA_{eq} \by,
\ee
where, $\bA_{eq}$ is GFDM equalization matrix which can be given as,
\be \label{eqn:gfdmeualizermatrix}
\bA_{eq}= \begin{cases}
\bA^{\rm H} ~ \text{for MF Equalizer} \\
  \bA^{-1} ~ \text{for ZF Equalizer} \\
  [ \mI+\AHA]^{-1} \bA^{\rm H}  ~ \\ \text{for biased MMSE  Equalizer}\\
 \matrixmmsebias{gfdm}^{-1} [ \mI+\AHA]^{-1} \bA^{\rm H} ~ \\ \text{for unbiased MMSE Equalizer},
 \end{cases}
\ee
where, $ \matrixmmsebias{gfdm}^{-1}$ is a diagonal bias correction matrix for GFDM-MMSE equaliser, where, 
\be \label{eqn:gfdmmatrixmmsebias}
\matrixmmsebias{gfdm}=diag\{[ \mI+\AHA]^{-1} \AHA\}.
\ee

\section{Low Complexity GFDM Transmitter} \label{sec:transmitter}
In this section, we present low complex GFDM transmitter. $\bA$ matrix is factorized into special matrices to obtained low complexity transmitter without incurring any assumptions related to GFDM parameters. In the following subsections, we will explain the design and implementation of the transmitter. 

\subsection{Low Complexity Transmitter Design} GFDM trasmitted signal is critcally sampled Inverse Descrete Gabor Transform (IDGT) of $\bd$ \cite{matthe_generalized_2014}. Using the IDGT matrix factorization given in \cite{stewart_computationally_1995}, the Modulation Matrix, $\bA$ can be given as,
  \begin{equation}
  \label{eqn:Afac}
  \begin{aligned}
  \bA&= ~~~~~~~~~~\bm G ~~~~~~~~~\times ~~~~~\diagdft{N} \\
  &={\begin{pmatrix}
     \bm \Psi _0 & \bm \Psi_{M-1}&\cdots & \bm \Psi_{1}\\
\bm \Psi_1 & \bm \Psi_0 &\cdots &\bm \Psi_2\\
 \vdots &\vdots& \ddots & \vdots\\
 \bm \Psi_{M-1} &\bm \Psi_{M-2}& \cdots & \bm \Psi_0
 \end{pmatrix}}   {\begin{pmatrix}
\bm W_N & & \\
 & \ddots & \\
&  & \\
  &  & \bm W_N
 \end{pmatrix}},
  \end{aligned}
  \end{equation}
where, $\bm \Psi_m=\text{diag}\{g[mN],g[mN+1], \ldots,g[mN+N-1]\}$ for $0\leq m\leq M-1$, is $N \times N$ diagonal matrix and $\bm W_N$ is $N \times N$ normalized IDFT matrix. 
The matrix $\bm G$ is block circulant matrix with diagonal blocks. $\bm G$ can be further factorized as\cite{mazancourt_inverse_1983},

\begin{equation}
\label{eqn:factorizationG}
\begin{aligned}
\bm G &= \bm F_b \bD \bm F_b^{\rm H},\\
\text{where},~~ \bm F_b &= \dftmat_M \kron \identity_N ~~~~\text{and}
\end{aligned}
\end{equation}
$\bD =\text{diag}\{\mathbb{D}_0,~\mathbb{D}_1,\ldots \mathbb{D}_{M-1}\}$, is a block diagonal matrix. The $q^{\rm th}$ block of $\bD$ ie. $\mathbb{D}_q$ can be obtained as $M$ point DFT of $\bm \Psi_m$ blocks and can be given as,
\begin{equation}
\label{eqn:evaluationDq}
\mathbb{D}_q=\sum_{m=0}^{M-1}{\omega^{-qm}\bm \Psi_m}, ~ ~ 0\leq q \leq M-1,
\end{equation}
where $\omega=\cexpM{}$. Since $\bm \Psi_m$'s are diagonal, $\mathbb{D}_q$'s are also diagonal and hence $\bD=diag\{\lambda(0),~\lambda(1),\cdots \lambda(MN-1)\}$, is also diagonal whose $r^{\rm th}$ diagonal value can be given as. 
\be \label{eqn:DiagonalvaluesofDmatrix}
\lambda(r)=\sum_{m=0}^{M-1}{g[mN+r \mod N]\omega^{-m\floor*{\frac{r}{N}}}},~0\leq r \leq MN-1.
\ee
Using (\ref{eqn:Afac}-\ref{eqn:factorizationG}), the modulation matrix $\bA$ can be given as,
\begin{equation}
\label{eqn:Modmatrix}
\bA=  \bm F_b \bD \bm F_b^{\rm H} \diagdft{N}.
\end{equation}

 Using (\ref{eqn:transmittedsignalmatrixform},\ref{eqn:Modmatrix}), the transmitted signal $\bx$ can be given as,
\begin{equation}
\label{eqn:transmittedsignal}
\bx=  \bm F_b \bD \bm F_b^{\rm H} \diagdft{N} \bd.
\end{equation}

\begin{lemma} \label{th:factorization G}
Let $\bS=diag\{s(0),~s(1)\cdots s(MN-1)\}$ be a diagonal matrix of size $MN\times MN$. The matrix $\bR=\bdft \bS \bdft^{\rm H}$ can be written as,
\begin{equation}
\bR= \permut^{\rm T} \diagdft{M} \bar{\bS} \diagdft{M}^{\rm H} \permut,
\end{equation}
where, $\diagdft{M} = diag\{\dftmat_M,~\dftmat_M \cdots \dftmat_M\}$ is a block diagonal matrix which holds $N$, $M$ order normalised IDFT matrix on its diagonal, $\permut$ is a subset of perfect shuffle permutation matrix, which can be defined as, $\permut =[p_{l,q}]~0\leq l,q\leq MN-1$, where the matrix element $p_{l,q}$ can be given as,
\begin{equation}
\label{eqn:permut}
p_{l,q}=\begin{cases} 1 ~\text{if}~ q=lN+\floor*{\frac{l}{M}}\\ 0 ~~ \text{otherwise}\end{cases}.
\end{equation} 
and $\bar{\bS}=diag\{\bar{s}(0),~\bar{s}(1)\cdots\bar{s}(MN-1)\}$ is a diagonal matrix which can be given as,
\begin{equation} \label{eqn:sr}
\bar{s}(r) = s((r \mod M)N+\floor*{\frac{r}{M}}), ~ \text{for} ~ 0\leq r \leq MN-1.
\end{equation}   
\end{lemma}

\begin{proof}
See Appendix~\ref{app:lem1}.
\end{proof}
\begin{theorem} \label{Th:modmat}
The GFDM modulation matrix $\bA$ can be given as,
\begin{equation}
\label{eqn:lowcomplextransmit}
\bA= \permut^{\rm T} \diagdft{M} \bar{\bD} \diagdft{M}^{\rm H} \permut \diagdft{N},
\end{equation}
where, $\bar{\bD} = diag\{\bar{\lambda}(0),~\bar{\lambda}(1)\cdots\bar{\lambda}(MN-1)\}$ is diagonal matrix, whose $r^{\rm th}$ element can be given as,
\begin{equation}
\label{eqn:transmitclosedformdbar}
\bar{\lambda}(r)=\sum_{m=0}^{M-1}{g[mN+\floor*{\frac{r}{M}}]\omega^{m(r \mod M)}}.
\end{equation}
\end{theorem}
\begin{proof}
See Appendix~\ref{app:Th1}.
\end{proof}
\begin{corollary}
Using Theorem~\ref{Th:modmat} and (\ref{eqn:transmittedsignalmatrixform}), GFDM transmitted signal, $\bx$ can be given as,
\be \label{corr:tx}
\bx= \permut^{\rm T} \diagdft{M} \bar{\bD} \diagdft{M}^{\rm H} \permut \diagdft{N} \bd.
\ee
\end{corollary}
\begin{lemma} \label{th:permutation}
Let $\tempvec=[\vartheta(0) ~ \vartheta(1) \cdots \vartheta(MN-1)]^{\rm T}$ be a $MN$ length complex valued vector. The vector, $\tilde{\tempvec}=\permut \tempvec= [\tilde{\vartheta}(0) ~ \tilde{\vartheta}(1) \cdots \tilde{\vartheta}(MN-1)]^{\rm T}$. The $i^{\rm th}$ element of the vector $\tilde{\tempvec}$ can be given as,
\begin{equation}
\label{eqn:Pv}
\tilde{\vartheta}(i) = \vartheta((i \mod M)N+\floor*{\frac{i}{M}}), ~ 0\leq i  \leq MN-1.
\end{equation}
The vector, $\bar{\tempvec}=\permut^{\rm T} \tempvec= [\bar{\vartheta}(0) ~ \bar{\vartheta}(1) \cdots \bar{\vartheta}(MN-1)]^{\rm T}$. The $i^{\rm th}$ element of the vector $\bar{\tempvec}$ can be given as,
\begin{equation}
\label{eqn:Ptv}
\bar{\vartheta}(i) = \vartheta((i \mod N)M+\floor*{\frac{i}{N}}), ~ 0\leq i \leq MN-1.
\end{equation}
\end{lemma}
\begin{proof}
Using (\ref{eqn:permut}) in $\tilde{\tempvec}=\permut \tempvec$ and $\bar{\tempvec}=\permut^{\rm T} \tempvec$, (\ref{eqn:Pv}) and (\ref{eqn:Ptv}) are obtained.
\end{proof}

\subsection{Low Complexity Transmitter Implementation} \label{sec:lowtx}
The low complexity transmitter can be obtained using Corollary~1 and Lemma~\ref{th:permutation}.  Fig.~\ref{fig:TX} presents low complexity transmitter implementation. 

The vector $\bm e= \diagdft{N} \bd$, can be obtained by $M$, $N$ point IFFT. The vector $\tilde{\bm e }=\permut \bm e$ can be obtained by shuffling the vector $\bm e$ using (\ref{eqn:Pv}).  The vector $\bm c= \diagdft{M}^{\rm H} \tilde{\bm e}$ can be obtained using $N$, $M$-point FFT's. Using (\ref{eqn:transmitclosedformdbar}), the matrix $\bar{\bm D}$, can be precomputed at the transmitter. The $MN$ length vector, $\bm z= \bar{\bm D} \bm c$, can be obtained by $MN$-point complex multiplication. The $MN$ length vector, $\tilde{\bx}= \diagdft{M} \bm z$ can be implemented using $N$, $M$-point IFFT. Finally, the transmitted signal, $\bx=\permut^{\rm T}\tilde{\bx}$ can be obtained by shuffling $\tilde{\bx}$ according to (\ref{eqn:Ptv}).
\begin{figure}[h]
\centering
 
  \includegraphics[width=\linewidth]{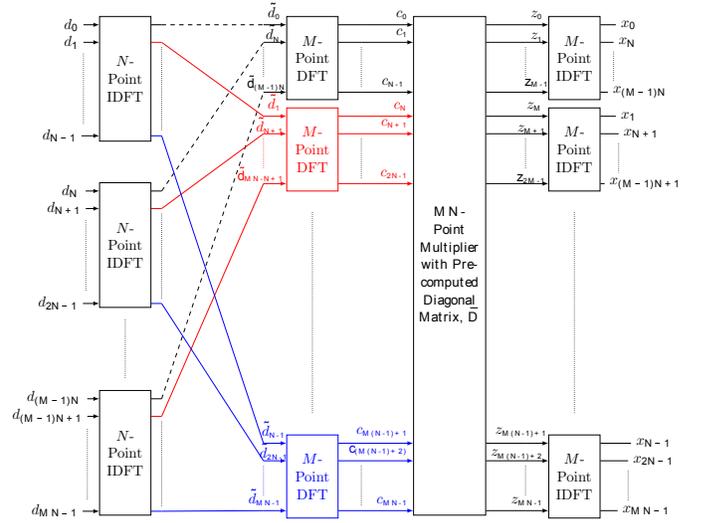}
  \caption{Low Complexity Implementation of GFDM Transmitter.}
  \label{fig:TX}
\end{figure}

  \begin{figure*} [t]
\centering
  \includegraphics[width=\linewidth]{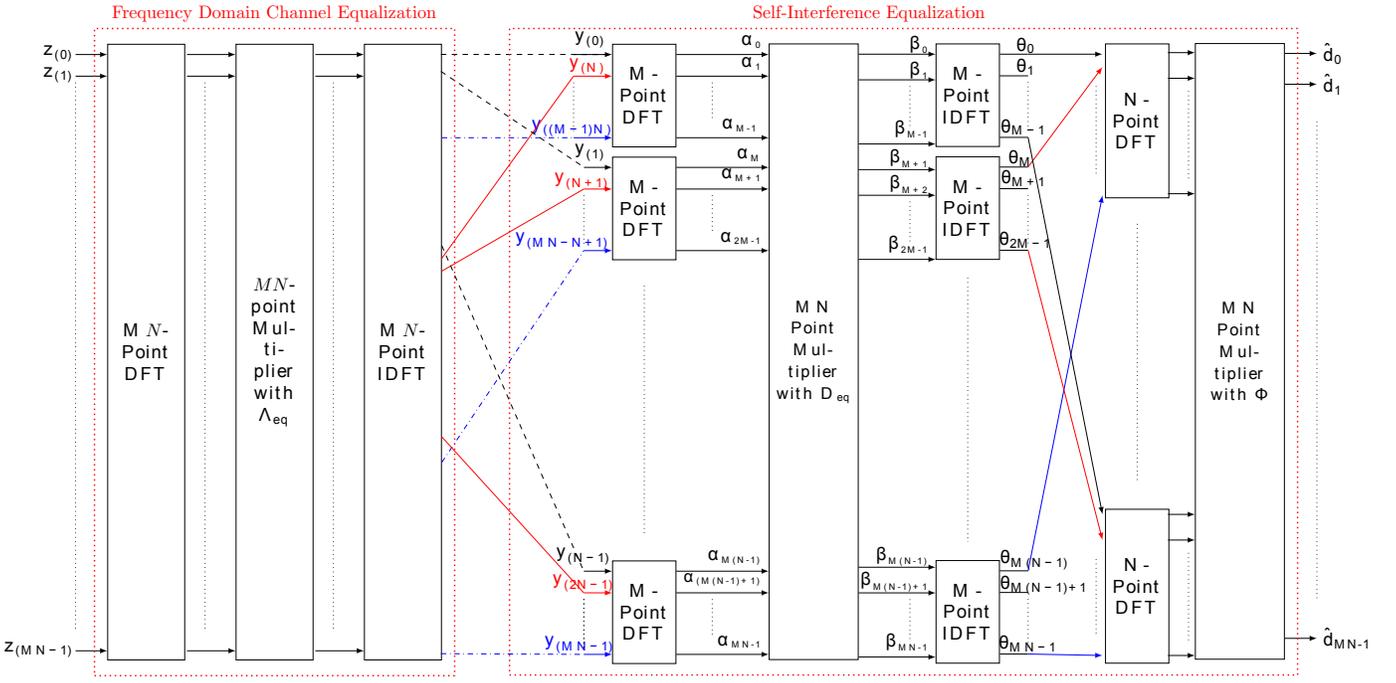}
  \caption{Low Complexity Implementation of GFDM Receiver in Multipath Fading Channel.}
  \label{fig:Rx}
\end{figure*}
\section{Low Complexity GFDM Reciever} \label{sec:reciever}
 In this section, we present the low complexity leaner GFDM receivers i.e. (1) MF (2) ZF and (3) Biased MMSE and (4) Unbiased MMSE. In the following subsections, we will show that using the factorization of $\bA$ given in (\ref{eqn:Afac}), receivers will low computational load can be designed. Additionally, our proposed receivers have unified implementation and will lead to similar computational complexity.   Our design does not make any assumption related to GFDM parameters, hence can achieve optimum performance. 
\subsection{Low Complexity GFDM Receiver Design}   
Receiver in AWGN channel is self-interference equalization. For multipath fading channel, channel equalization is followed by self-interference equalization. 
Theorem~\ref{Th:Aeq} relates to unified low complexity GFDM linear receivers.    
 \begin{theorem} \label{Th:Aeq}
GFDM equalization matrix $\bA_{eq}$ can be written in a unified manner as,
\begin{equation} \label{eqn:unifiedequalizationmatrix}
\bA_{eq} = \matrixmmsebias{}  \diagdft{N}^{\rm H} \permut^{\rm T} \diagdft{M} \bD_{eq} \diagdft{M}^{\rm H} \permut,
\end{equation}
where, $\bD_{eq}$ is a diagonal $MN$-order matrix, which can be given as,
\begin{equation} \label{eqn:Deqmatrix}
\bD_{eq} = \begin{cases}
\diagonalgfdm^{\rm H} ~ \text{for MF} \\
\diagonalgfdm^{-1} ~ \text{for ZF} \\
[\mI +abs\{\diagonalgfdm\}^2]^{-1} \diagonalgfdm^{\rm H} \\
\text{for unbiased and biased MMSE}
\end{cases}
\end{equation}
 and $\matrixmmsebias{}= \matrixmmsebias{gfdm}^{-1}$ for unbiased MMSE and $\matrixmmsebias{}=\bI_{MN}$ for other equalizers. Further, $\matrixmmsebias{gfdm}$ can be given as,
 \begin{equation} \label{eqn:th:gfdmmmsebiasmatrix}
 \matrixmmsebias{gfdm} = \frac{1}{MN} \sum_{r=0}^{MN-1}{\frac{|{\lambda}_r|^2}{|{\lambda}_r|^2+\snri}} \bI_{MN}.
 \end{equation}
 \end{theorem}
 \begin{proof}
This theorem can be proved using the factorization of $\bA$ in (\ref{eqn:Afac}), properties of Kronecker product and properties of unitary matrices. For complete proof,  see Appendix~\ref{app:proofTh2}.
 \end{proof}
 \begin{corollary} \label{corr:twostageestdata}
 Using Theorom~\ref{Th:Aeq} and (\ref{eqn: rec}-\ref{eqn:gfdmmatrixmmsebias}), the estimated data, $\hat{\bd}$, can be given as,
 \begin{equation} \label{eqn:lowcomplextwostageestimateddata}
 \hat{\bd} =\begin{cases} \matrixmmsebias{}  {\diagdft{N}^{\rm H}} \permut^{\rm T} \diagdft{M} \bD_{eq} \diagdft{M}^{\rm H} \permut \z ~\text{for AWGN Channel}\\
\matrixmmsebias{}  {\diagdft{N}^{\rm H}} \permut^{\rm T} \diagdft{M} \bD_{eq} \diagdft{M}^{\rm H} \permut  \dftmat_{MN} \channelfreqmatrix_{eq} \dftmat_{MN}^{\rm H} \z \\~ \text{for Multipath Fading Channel} \end{cases}
 \end{equation}
 \end{corollary}

\subsection{Low Complexity Receiver Implementation} \label{sec:lowrx}
 The low complexity structure of GFDM self-interference cancellation can be obtained by using Corollary~\ref{corr:twostageestdata} and Lemma~\ref{th:permutation}. Low complexity two stage receiver implementation can be understood in the light of Fig.~\ref{fig:Rx}.
 
 \subsubsection{Channel Equalization}
 To implement $\by^1=   \channelfreqmatrix_{eq} \dftmat_{MN}^{\rm H}\z$, $MN$-point FFT of $\z$ is  multiplied with  $ \channelfreqmatrix_{eq}$. Finally, we take $MN$-point IFFT of $\by^1$ to implement $\by=\dftmat_{MN}\by^1$.
\subsubsection{Self-interference Equalization} 
  The vector $\tilde{\by}=\permut \by$, can be obtained by shuffling the $\by$ vector using (\ref{eqn:Pv}).  The $MN \times 1$ vector $\bsy \alpha= \diagdft{M}^{\rm H} \tilde{\by}$ can be implemented by using $N$, M-point IFFT's. The vector $\bsy {\alpha}$ is then multiplied to the diagonal matrix $\bD_{eq}$ to obtain $\bsy {\beta}$. The vector $\bsy \theta=\diagdft{M}\bsy{\beta}$ can be implemented using $N$, $M$-point FFTs. The vector, $\tilde{\bsy \theta}=\permut^{\rm T} \bsy{\theta}$, can be implemented by shuffling the $\bsy{\theta}$ vector using (\ref{eqn:Ptv}). Now, the vector, $\bar{\bd}=\diagdft{N}\tilde{\bsy{\theta}}$ can be implemented using $M$, $N$-point FFTs. Finally, $\hat{\bd}=\matrixmmsebias{}\bar{\bd}$ can be obtained by using $MN$-point multiplier.

\section{Complexity Computation}\label{sec:complexitycomputation}
In this section, we present the computational complexity of GFDM transmitter and receivers proposed in this work. We calculate the complexity in terms of the total number of real multiplications and real additions. $FFT_{(.)}$ and $IFFT_{(.)}$ denote $(.)$-point FFT and IFFT respectively.

For evaluation of computational complexity, we compute the flops required where one flop indicates one real multiplication or one real addition. Flops needed for one complex multiplication, division, addition, conjugate and modulus square are  6, 6, 2, 2 and 3 respectively. 

We consider value of $N$ and $M$ to be a power of two. To enable this we use modified raised cosine pulse derived for even $M$ values in \cite{nimr_optimal_2017}. In the light of Sec.~\ref{sec:lowtx} and Sec.~\ref{sec:lowrx}, it is clear that our low complexity transceiver is implemented using $N$, $M$ and $MN$ point FFT and IFFT algorithms. The choice of FFT/IFFT algorithm is a critical aspect for complexity computation. Let us consider, the flops required to compute $N$, $M$ and $MN$ point FFT/IFFT are $\flopdft{N}$, $\flopdft{M}$ and $\flopdft{MN}$.  It will be shown in further subsections that complexity of our proposed transceiver is given in terms of $N\times \flopdft{M}$ and $M\times\flopdft{N}$. Value of $N$ and $M$ can be small as well as large, however it is unlikely that both $N$ and $M$ are small simultaneously. Hence, flops required for FFT/IFFT of small inputs are also important for transceiver implementation.  In Appendix~\ref{sec:app:fft}, it is shown that Winograd's FFT \cite{winograd_computing_1978}  requires lesser flops than radix-2 and split-radix FFT \cite{duhamel_fast_1990} when $N,~M$ is small ($\leq$ 19) . Hence, when $N$ or $M$ is small we choose Winograd's small FFT algorithm. For input values of [2 4 8 16], Winograd's small FFT requires [4 12 34 92] flops. When $M,~N ~\geq 32$, we implement split-radix algorithm to implement $N$ or $M$ point FFT/IFFT. We consider that $MN$-point FFT/IFFT is also implemented using split-radix algorithm. Flops to compute $X$ point FFT/IFFT using split-radix algorithm is $4X\log_2X-6X+8$ \cite{duhamel_fast_1990}.

 
 Assumptions to design receiver in \cite{wei_fast_2016} are incompatible with ones in \cite{farhang_low_2015, matthe_precoded_2016,lin_orthogonality_2015,farhang-boroujeny_derivation_2015}. We have adhered to conditions in \cite{farhang_low_2015, matthe_precoded_2016,lin_orthogonality_2015,farhang-boroujeny_derivation_2015}. Because of contrary assumptions, comparison of our tranceiver with \cite{wei_fast_2016} is beyond the scope of our work. 
\subsection{Transmitter}
 As evident from Sec~\ref{sec:transmitter} and Fig.~\ref{fig:TX}, the transmitter can be implemented using $M$ numbers of $FFT_N$, $N$ numbers of $FFT_{M}$, $N$ numbers of $IFFT_M$ and  $MN$ complex divisions. Table~\ref{tab:complexity} presents the total number of complex multiplication needed to implement different transmitter structures.
 \begingroup
\setlength{\tabcolsep}{10pt} 
\renewcommand{\arraystretch}{1.5} 
 \begin{table}[h]
\caption{Computational Complexity of different GFDM Transmitter Implementations. $\flopdft{N}$ and $\flopdft{M}$ are flops required to compute N and M point FFT respectively.}
\centering
\label{tab:complexity}
\begin{tabular}{| >{\centering\arraybackslash}m{0.35\linewidth}| >{\centering\arraybackslash}m{0.45\linewidth}|}
\hline
\textbf{Structure} & \textbf{Number of Flops}\\ \hline \hline
Transmitter in \cite{michailow_generalized_2012}  & $M\times\flopdft{N}+2N\times \flopdft{M}+6MNL $, where $1<L\leq N$ \\  \hline 
 Transmitter in \cite{farhang-boroujeny_derivation_2015} & $M \times \flopdft{N}+4M^2N$ \\ \hline
 Transmitter in \cite{lin_orthogonality_2015,matthe_precoded_2016,farhang_low-complexity_2016} & $M \times \flopdft{N}+3M^2N+2(M-1)N$ \\ \hline
  Our Proposed Transmitter  & $M\times\flopdft{N}+2N\times \flopdft{M}+6MN $\\ \hline
  OFDM Transmitter & $M\times \flopdft{N}$\\ \hline
\end{tabular}

\end{table}

Complexities presented in Table~\ref{tab:complexity} are plotted in Fig.~\ref{fig:complexity:flops:N} for $N=16$ and $M\in[2,~1024]$ and Fig.~\ref{fig:complexity:flops:M} for $M=16$ and $N\in [2,~1024]$.  In this work, since we assume any arbitrary pulse shape. Hence, for fairness of comparison, we take $L=N$ for transmitter in \cite{michailow_generalized_2012}. It is observed that for $M<8$, transmitter in  \cite{farhang_low-complexity_2016} achieves the lowest complexity. For instace for $M=4$, our transmitter requires 50 percent more computational load than one in \cite{farhang_low-complexity_2016}.  For $M\geq8$, our proposed transmitter has the lowest computational complexity.  It is worth mentioning that complexity of transmitter in \cite{farhang_low-complexity_2016} is quadratic with $M$ whereas complexity of transmitter in \cite{michailow_generalized_2012} is quadratic with $N$. However, complexity of our proposed transmitter is log-linear with $M$ as well as $N$. Hence, when $M$ is large, our transmitter achieves significant complexity gain over transmitter in \cite{farhang_low-complexity_2016}. In the same way, when $N$ is large our transmitter acheives significant complexity gain over transmitter in \cite{michailow_generalized_2012}. For instance, when $M=1024$ and $N=16$, our transmitter is 100 times lesser complex than one in \cite{farhang_low-complexity_2016} and 3 times lesser complex than one in \cite{michailow_generalized_2012}. When $N=1024$ and $M=16$, our transmitter is 100 times lesser complex than one in \cite{michailow_generalized_2012} and 25 percent lesser complex than one in \cite{farhang_low-complexity_2016}. When compared with OFDM, our proposed transmitter has 2 to 10 times higher computational load.  Comparative complexity of GFDM transmitter with OFDM transmitter increases with $M$ and decreases with $N$. For instance, when $N=1024$ and and $M=16$, our transmitter has two higher complexity than OFDM. Whereas, when $M=1024$ and $N=16$, our transmitter is 6 times more complex than OFDM. It can be concluded that our proposed transmitter provides low computational load for flexible GFDM transmitter which may take arbitrary values of $M$, $N$ and arbitrary pulse shape.

\begin{figure}[h]
\begin{subfigure}{0.49\textwidth}
\includegraphics[width=\linewidth]{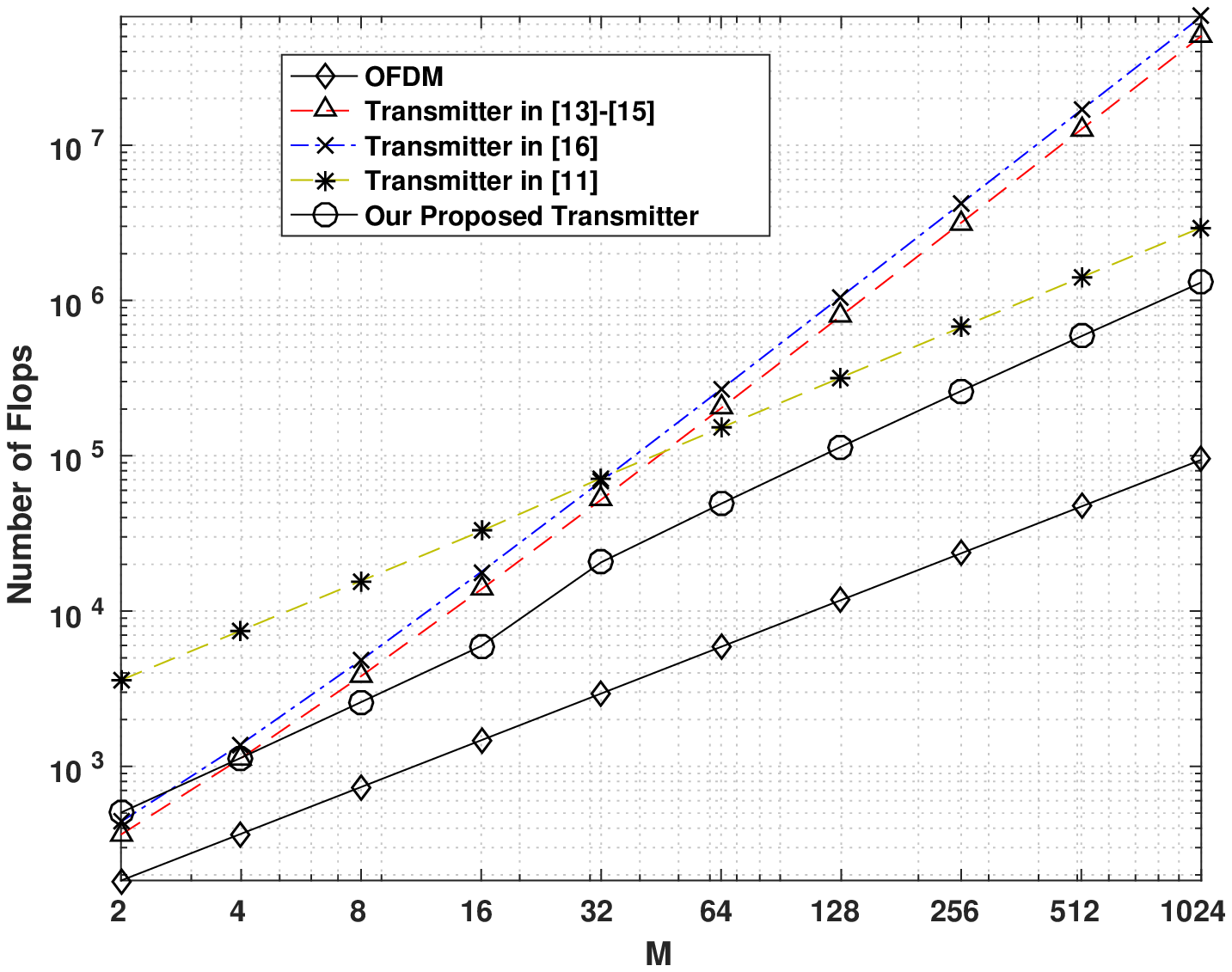}
\caption{ $N=16$, $M\in [2 ,~1024]$  }
\label{fig:complexity:flops:N}
\end{subfigure}
~
\begin{subfigure}{0.49\textwidth}
\includegraphics[width=\linewidth]{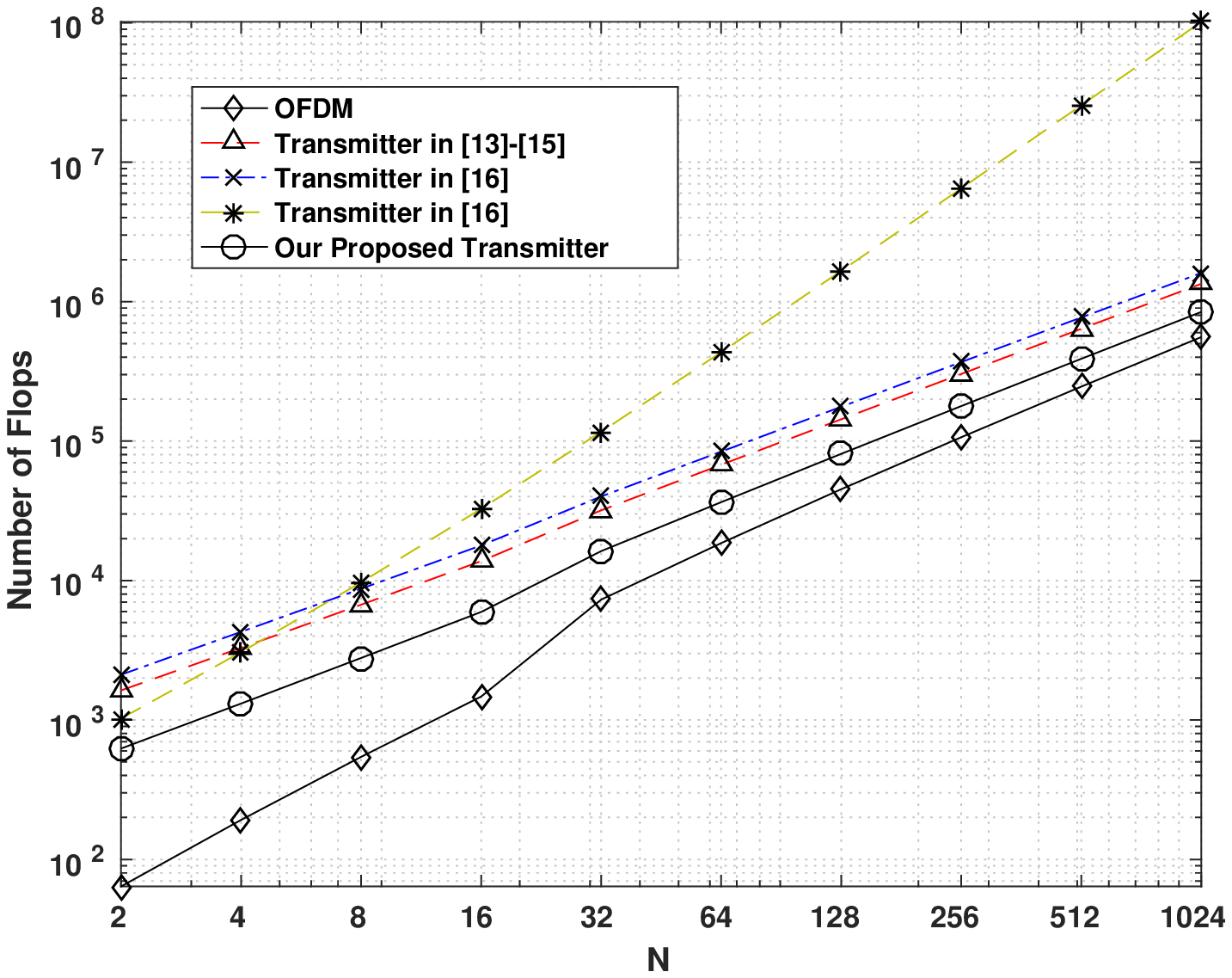}
\caption{$M=16$, $N\in [2,~ 1024]$ }
\label{fig:complexity:flops:M}
\end{subfigure}
\caption{Computational Complexity of different Transmitters.}
\label{fig:complexitytx}
\end{figure}

\subsection{Receiver} \label{sec:complexity:rx}
In this section we discuss the computation complexity of our proposed receivers for AWGN as well as multipath fading channel. 
\subsubsection{AWGN Channel} \label{sec:complexity:rx:awgn}
For AWGN channel, channel equalization is not needed. So, we will discuss the computational complexity of self-interference equalizer in this section.
As evident from Sec~\ref{sec:reciever} and Fig.~\ref{fig:Rx}, the receiver can be implemented using $M$ numbers of $FFT_N$, $N$ numbers of $FFT_M$, $N$ numbers of $IFFT_M$ and $MN$ complex multiplications. The $\bD_{eq}$ matrix can be precomputed for ZF and MF receiver. But, for MMSE receiver, $\bD_{eq}$ can be computed using $MN$ real additions, $2MN$ real multiplications and $2MN$ real divisions when the matrix $\bar{\bD}^{\rm H}$ and $abs\{\bar{\bD}\}^2$ are precomputed.  For MF, ZF and biased-MMSE receiver multiplication with $\matrixmmsebias{gfdm}$ is trivial and does not require any flops to implement. For Biased-MMSE receiver, the $\matrixmmsebias{gfdm}^{-1}$ can be computed using $2MN-1$ real additions, $MN+1$ real divisions and $MN$ modulus square operation.
 Computational complexity of different receivers in AWGN channel is given in Table~\ref{tab:complexity:awgn}. 
\begingroup
\setlength{\tabcolsep}{10pt} 
\renewcommand{\arraystretch}{1.5} 
\begin{table}[h]
\caption{Computational Complexity of Different Receivers in AWGN Channel. $\flopdft{MN}$, $\flopdft{N}$ and $\flopdft{M}$ are flops required to compute NM, N and M point FFT respectively.}
\label{tab:complexity:awgn}
\centering
\begin{tabular}{| >{\centering\arraybackslash}m{0.35\linewidth}| >{\centering\arraybackslash}m{0.5\linewidth}|}
\hline
\textbf{Structure} & \textbf{Number of Flops}\\ \hline\hline \vspace{0.1cm}
OFDM  & $M\times \flopdft{N}$ \\ \hline

 ZF/MF Receiver in \cite{farhang_low-complexity_2016,matthe_precoded_2016} & $M\times\flopdft{N}+3M^2N+2(M-1)N$ \\ \hline \vspace{0.1cm}
 ZF/MF Receiver in \cite{michailow_generalized_2014} & $2\times \flopdft{MN}+2N\times\flopdft{M}+6MNL$ , where $1<L\leq N$  \\ \hline \vspace{0.1cm}
 MMSE Receiver in \cite{farhang_low-complexity_2016,matthe_precoded_2016} & $M\times\flopdft{N}+12M^2N+9MN$ \\ \hline \vspace{0.1cm}
SIC Receiver in \cite{gaspar_low_2013} & $2\times\flopdft{MN}+2N \times \flopdft{M}+6LMN+I(4N\times \flopdft{M}+6MN)$ \\ \hline
  Proposed ZF/MF Receiver & $M\times \flopdft{M}+2N\times\flopdft{M}+6MN$  \\ \hline 
  Proposed biased MMSE Receiver & $M\times \flopdft{M}+2N\times\flopdft{M}+11MN$ \\ \hline
   Proposed unbiased MMSE Receiver & $M\times \flopdft{M}+2N\times\flopdft{M}+17MN$\\ \hline
%
\end{tabular}
\end{table}
\endgroup
\begin{figure}[h]
\begin{subfigure}{0.49\textwidth}
\includegraphics[width=\linewidth]{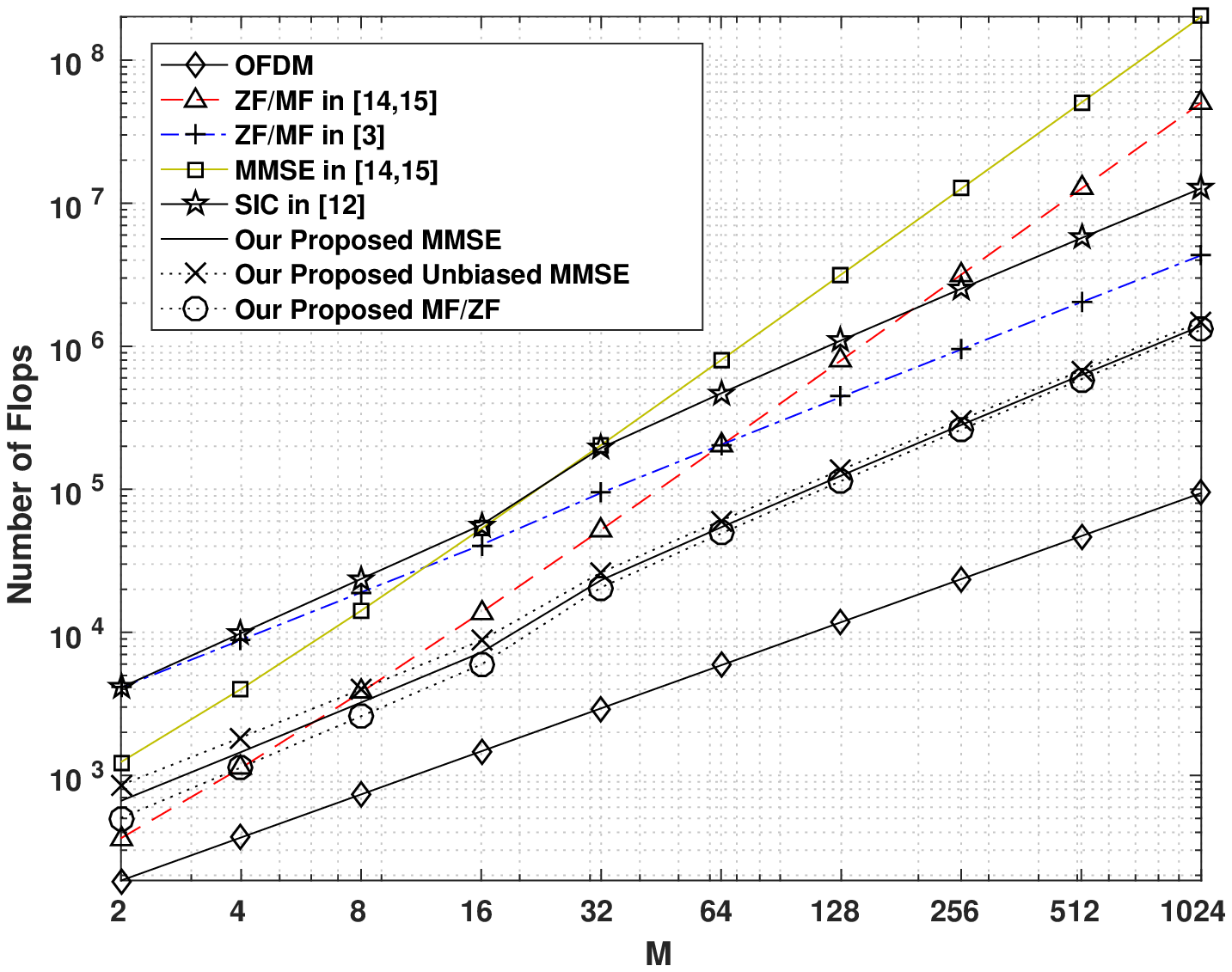}
\caption{ $N=16$, $M\in [2 ,~1024]$  }
\label{fig:complexity:awgn:flops:N}
\end{subfigure}
~
\begin{subfigure}{0.49\textwidth}
\includegraphics[width=\linewidth]{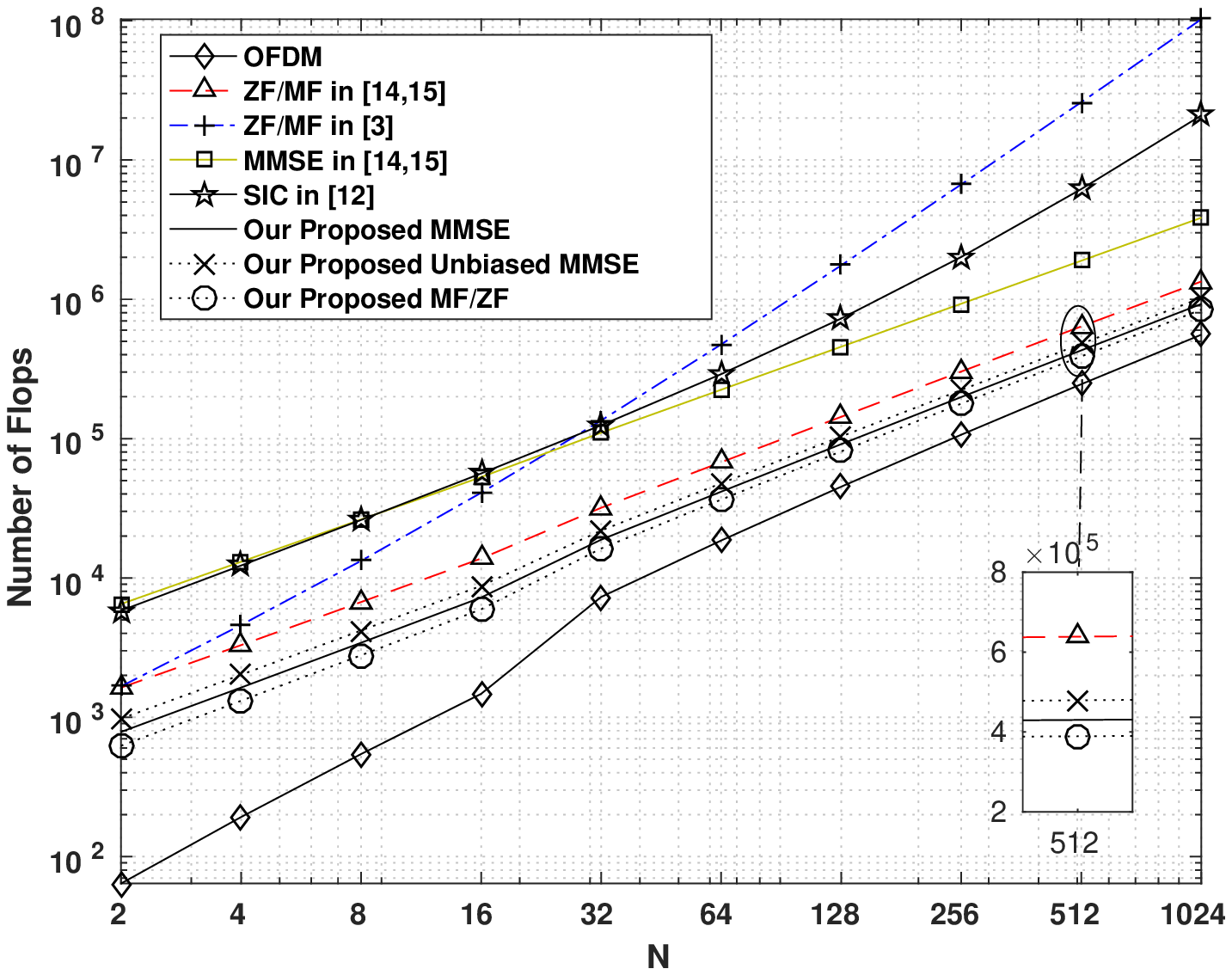}
\caption{$M=16$, $N\in [2,~ 1024]$ }
\label{fig:complexity:awgn:flops:M}
\end{subfigure}
\caption{Computational Complexity of different Receivers in AWGN Channel.}
\label{fig:complexityrxawgn}
\end{figure}

Complexities presented in Table~\ref{tab:complexity:awgn} are plotted in Fig.~\ref{fig:complexity:awgn:flops:N} for $N=16$ and $M\in[2,~1024]$ and Fig.~\ref{fig:complexity:awgn:flops:M} for $M=16$ and $N\in [2,~1024]$.  In this work, since we assume any arbitrary pulse shape. Hence, for fairness of comparison, we take $L=N$ for recievers in \cite{michailow_generalized_2014,matthe_generalized_2014,gaspar_low_2013}.
Using the results in \cite{gaspar_low_2013}, we consider $I=8$ for SIC receiver in \cite{gaspar_low_2013}. It is observed that Our proposed MF/ZF, biased MMSE and Unbiased MMSE have similar complexities since biased and unbiased MMSE requires only $5MN$ and $11MN$ additional flops over ZF. When $M$ is small, our proposed MF/ZF receiver has similar complexity to MF/ZF receiver in \cite{farhang_low-complexity_2016}. As $M$ increases our proposed MF/ZF receiver  attains significant complexity gain over ones in \cite{farhang_low-complexity_2016}. 
Our proposed MMSE receiver has the lowest computation load and achievs complexity reduction of 3 to 300 times in comparision to the ones in \cite{michailow_generalized_2014,matthe_generalized_2014,farhang_low-complexity_2016,matthe_precoded_2016}. 
 It is worth mentioning that complexity of MF/ZF/MMSE Receiver in \cite{farhang_low-complexity_2016} is quadratic with $M$, complexity of ZF receiver in \cite{michailow_generalized_2012} is quadratic with $N$ and complexity of MMSE receiver in \cite{michailow_generalized_2014,matthe_generalized_2014} is quadratic with $M$ as well as $N$. However, complexity of our proposed recievers are log-linear with $M$ as well as $N$. Hence, when $M$ is large, our MF/ZF/MMSE receiver achieves significant complexity gain over MF/ZF/MMSE receiver in \cite{farhang_low-complexity_2016} and MMSE receiver in \cite{michailow_generalized_2014,matthe_generalized_2014}. In the same way, when $N$ is large our MF/ZF/MMSE receiver acheives significant complexity gain over MF/ZF/MMSE receiver in \cite{michailow_generalized_2012} and MMSE receiver in \cite{michailow_generalized_2014,matthe_generalized_2014}. For instance, when $M=1024$ and $N=16$, our MF/ZF receiver and MMSE receiver are 100 and 300 times lesser complex than ones in \cite{farhang_low-complexity_2016}. When $N=1024$ and $M=16$, our MF/ZF receiver and MMSE are respectively 100  and 300 times lesser complex than ones in \cite{michailow_generalized_2012,michailow_generalized_2014,matthe_generalized_2014}. 
In comparison to SIC receiver in \cite{gaspar_low_2013}, our proposed receivers are 100 to 200 times lesser complex.   It can be concluded that our proposed receivers attain significant complexity reduction as compared to receivers in \cite{michailow_generalized_2014,matthe_generalized_2014,michailow_generalized_2012,
 farhang_low-complexity_2016,gaspar_low_2013}. 
 
  When compared with OFDM, our proposed MF/ZF/MMSE receiver has 2 to 10 times higher computational load.  Comparative complexity of our proposed GFDM MF/ZF/MMSE receiver with OFDM receiver increases with $M$ and decreases with $N$. For instance, when $N=1024$ and and $M=16$, our  receivers have two times higher complexity than OFDM receiver. Whereas, when $M=1024$ and $N=16$, our receivers are 6 times more complex than OFDM. 

\subsubsection{Multipath Fading Channel}
As discussed in Sec.~\ref{sec:reciever}, receiver for multipath fading channel requires channel equalization in addition to self-interference equalization. Hence to compute complexity of receiver in multipath fading channel, channel equalization complexity needs to be added to complexity required for self-interference equalization which is computed in Sec.~\ref{sec:complexity:rx:awgn}. As discussed in Sec.~\ref{sec:lowrx}, channel equalization can be implemented using one FFT$_{MN}$, one IFFT$_{MN}$ and one MN-point complex multiplication when $\channelfreqmatrix_{eq}$ is known. This requires $\flopdft{MN}+6MN$ flops. For ZF-FDE, $\channelfreqmatrix_{eq}$ does not need any flops. Using (\ref{eqn:systemmodel:channelequalized}), $\channelfreqmatrix_{eq}$ for MMSE receiver can be computed using $MN$-point modulus square to compute $|\channelfreqmatrix|^2$ which requires $3MN$ flops, one $MN$-point real adder which requires $MN$ flops and one $MN$-point real and complex divisions which requires  $3MN$ flops. Thus, MMSE FDE needs extra $7MN$ Flops over ZF FDE.  Computational complexity of different receivers in multipath fading channel is provided in Table~\ref{tab:complexity:multipath}.  
\begin{table*}[h]
\caption{Computational Complexity of Different Receivers in Multipath Fading Channel.  $\flopdft{MN}$, $\flopdft{N}$ and $\flopdft{M}$ are flops required to compute NM, N and M point FFT respectively.}
\label{tab:complexity:multipath}
\centering
\begin{tabular}{| >{\centering\arraybackslash}m{0.2\linewidth}| >{\centering\arraybackslash}m{0.35\linewidth}| >{\centering\arraybackslash}m{0.35\linewidth}|}
\hline
\textbf{Structure} & \textbf{Number of Flops for ZF FDE}&  \textbf{Number of Flops for MMSE FDE}\\ \hline \hline \vspace{0.1cm}
OFDM & $M\times \flopdft{N}+6MN$ & $M\times \flopdft{N}+13MN$ \\ \hline
 ZF/MF Receiver in \cite{farhang_low-complexity_2016,matthe_precoded_2016} & $2\times \flopdft{MN}+M\times\flopdft{N}+3M^2N+8MN-2N$ & $2\times \flopdft{MN}+M\times\flopdft{N}+3M^2N+15MN-2N$ \\ \hline \vspace{0.1cm}
 ZF/MF Receiver in \cite{michailow_generalized_2014} & $4\times\flopdft{MN}+2N\times\flopdft{M}+6LMN+I(4N\times\flopdft{M}+6MN)+6MN$, where $1<L\leq N$ & $4\times\flopdft{MN}+2N\times\flopdft{M}+6LMN+I(4N\times\flopdft{M}+6MN)+13MN$ \\ \hline \vspace{0.1cm}
  MMSE Receiver in \cite{farhang_low-complexity_2016,matthe_precoded_2016} & $2\times \flopdft{MN}+M\times\flopdft{N}+12M^2N+15MN$  &$2\times \flopdft{MN}+M\times\flopdft{N}+12M^2N+22MN$  \\ \hline \vspace{0.1cm}
  SIC Receiver in \cite{gaspar_low_2013} & $4\times\flopdft{MN}+2N \times \flopdft{M}+6LMN+I(4N\times \flopdft{M}+6MN)+6MN$   & $4\times\flopdft{MN}+2N \times \flopdft{M}+6LMN+I(4N\times \flopdft{M}+6MN)+13MN$   \\ \hline
  Proposed ZF/MF Receiver & $2\times \flopdft{MN}+M\times\flopdft{M}+2N\times\flopdft{M}+12MN$ & $2\times \flopdft{MN}+M\times\flopdft{M}+2N\times\flopdft{M}+19MN$ \\ \hline
  Proposed biased MMSE Receiver & $2\times \flopdft{MN}+M\times\flopdft{M}+2N\times\flopdft{M}+17MN$ & $2\times \flopdft{MN}+M\times\flopdft{M}+2N\times\flopdft{M}+24MN$\\ \hline
  Proposed unbiased MMSE Receiver & $2\times \flopdft{MN}+M\times\flopdft{M}+2N\times\flopdft{M}+23MN$ & $2\times \flopdft{MN}+M\times\flopdft{M}+2N\times\flopdft{M}+30MN$ \\ \hline
%
\end{tabular}

\end{table*}

\begin{figure}[h]
\begin{subfigure}{0.49\textwidth}
\includegraphics[width=\linewidth]{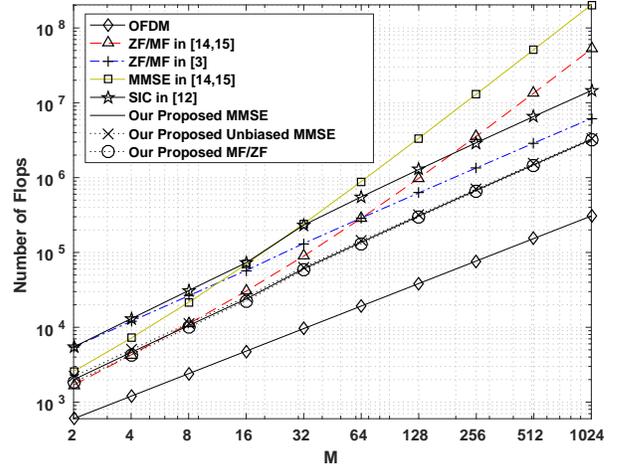}
\caption{$N=16$, $M\in [2,~ 1024]$ }
\label{fig:complexityfading:flops:M}
\end{subfigure}
~
\begin{subfigure}{0.49\textwidth}
\includegraphics[width=\linewidth]{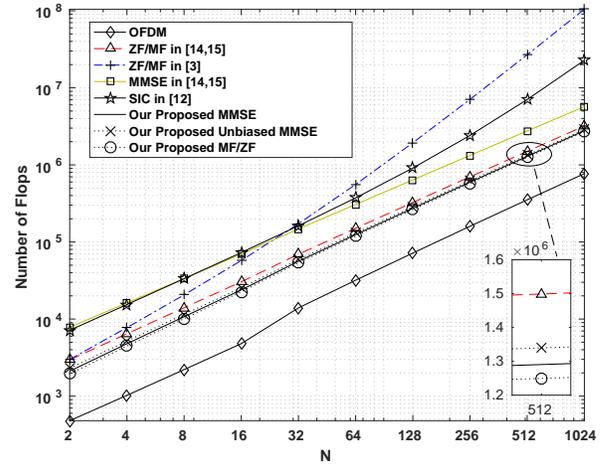}
\caption{ $M=16$, $N\in [2 ,~1024]$  }
\label{fig:complexityfading:flops:N}
\end{subfigure}
\caption{Computational Complexity of different Receivers in multipath fading channel with MMSE FDE.}
\label{fig:complexityfading}
\end{figure}
Complexities presented in Table~\ref{tab:complexity:multipath} are plotted in Fig.~\ref{fig:complexityfading:flops:N} for $N=16$ and $M\in[2,~1024]$ and Fig.~\ref{fig:complexityfading:flops:M} for $M=16$ and $N\in [2,~1024]$ for MMSE FDE\footnote{Since MMSE FDE requires only $7MN$ additional flops over ZF FDE. Results for ZF FDE will also be similar.}. It is worthwhile to note that complexity required for MMSE FDE will be added to all low complexity ZF/MMSE self-interference equalizers (our proposed and ones in \cite{michailow_generalized_2014,matthe_generalized_2014,michailow_generalized_2012,
 farhang_low-complexity_2016,gaspar_low_2013} ). Hence, comparative complexities of GFDM receivers will be similar as in the case of AWGN channel (see Sec.~\ref{sec:complexity:rx:awgn}). However, comparative complexity of GFDM receivers to OFDM receiver will increase as OFDM saves complexity in channel equalization. For instance, when $N=1024$ and and $M=16$, our receivers have 2.7 times higher complexity than OFDM. Whereas, when $M=1024$ and $N=16$, our receivers are 6.5 times more complex than OFDM.

It can be concluded that our proposed ZF receiver in AWGN as well as multipath fading channel is around 3 times simpler than ZF receivers in \cite{michailow_generalized_2014,matthe_generalized_2014,michailow_generalized_2012,
 farhang_low-complexity_2016}. Our Proposed Unbiased MMSE receiver in AWGN as well as multipath fading channel is 3 to 300 times simpler than biased MMSE receiver in \cite{michailow_generalized_2014,matthe_generalized_2014,michailow_generalized_2012,
 farhang_low-complexity_2016}. Our receivers retain low computational load for arbitrary value of $M$, $N$ and pulse shape. Our receivers also retain the optimal ZF and MMSE performance since they are direct. We will investigate performance optimality in detail in the next section.


\section{Performance of Low Complexity GFDM Transceiver}
In this section, we present Bit Error Rate (BER) performance of our proposed low complexity transceiver. As discussed earlier our proposed transceiver does not make any assumption  related to GFDM parameters as well as provides lowest computational load for arbitrary values of $M$, $N$ and pulse shape.  We present BER performance of our low proposed receiver in AWGN as well as multipath fading channel.  We compare the performance of our proposed receivers with the direct implementation of respective receivers given in Sec.~\ref{sec:systemmodel}. Simulation parameters are provided in Table~\ref{tab:simu:para:mmse}. We consider a system bandwidth of 1.92 MHz. We test our system for two cases namely (i) CaseI : $N=128$ and $M=8$ denoting a system where value of $N$ is high $M$ is low as well as sub-carrier bandwidth is low at 15 KHz and (ii) Case II : $N=8$ and $M=128$ denoting a system where value of $M$ is high $N$ is low as well as sub-carrier bandwidth is high at 240 KHz. Each point in our BER curve is calculated for $10^7$ transmission bits.
\begin{table} [h]
\caption{Simulation Parameters}
\label{tab:simu:para:mmse}
\centering
\begin {tabular}{| >{\centering\arraybackslash}m{0.29\linewidth}| >{\centering\arraybackslash}m{0.29\linewidth}|m{0.19\linewidth}|}
\hline
\textbf{Parameters} & \textbf{Case I} & \textbf{Case II} \\ \hline \hline
Number of Sub-carriers $N$ & 128 & 8\\ \hline
Number of Timeslots $M$ & 8 & 128 \\ \hline
Sub-carrier Bandwidth & 15 KHz & 240 KHz\\ \hline
\end{tabular}

\begin{tabular}{| >{\centering\arraybackslash}m{0.29\linewidth}| >{\centering\arraybackslash}m{0.56\linewidth}|}
Number of Sub-carriers for OFDM & 128 \\ \hline
Mapping & 16 QAM  \\ \hline
Pulse shape & modified RC \cite{nimr_optimal_2017} with ROF = 0.1 or 0.9 \\ \hline
Channel & AWGN or ETU\cite{series2009guidelines}\\ \hline
Carrier Frequency & 2.4 GHz \\ \hline
Maximum Doppler shift for multipath channel & 100 Hz \\ \hline

RMS delay Spread for multipath channel & 1 $\mu$ sec \\ \hline
Coherence Bandwidth for multipath channel & 20 KHz \\ \hline
Channel Equalization for multipath channel & MMSE FDE \\ \hline 
\end{tabular}

\end{table}

Next, we discuss about the parameter $M$. Let the system bandwidth remain constant while we discuss about changing $M$. There are two ways to change $M$ (a) Keeping the symbol duration  (or sub-carrier bandwidth) same and (b) Keeping the block duration (or value of $NM$) same. Without any loss of generality, in this article we increase $M$ while keeping the block duration constant. As $M$ is increasing $N$ is decreasing by the same amount i.e. sub-carrier bandwidth is increasing or symbol duration is decreasing while keeping the block length same. This way latency of the system does not change. One might ask the question about it's consequences on the performance especially so in fading channel. The design criteria to enable FDE in GFDM is different than that of in OFDM. In OFDM, sub-carrier bandwidth $\delta f_{\rm OFDM} < B_c$, where $B_c$ is coherence bandwidth of multipath channel. In case of GFDM, $\delta f_{\rm GFDM} < M \times B_c$ since the frequency resolution for FDE in case of GFDM is $\frac{1}{NM}$. So, in the light of aforementioned discussion, if we keep the value $NM$ constant, usage of FDE remains valid for GFDM. Additionally, PAPR will decrease as $M$ increases (since $N$ decreases).

\subsection{AWGN Channel}
\begin{figure}[h]
\includegraphics[width=\linewidth]{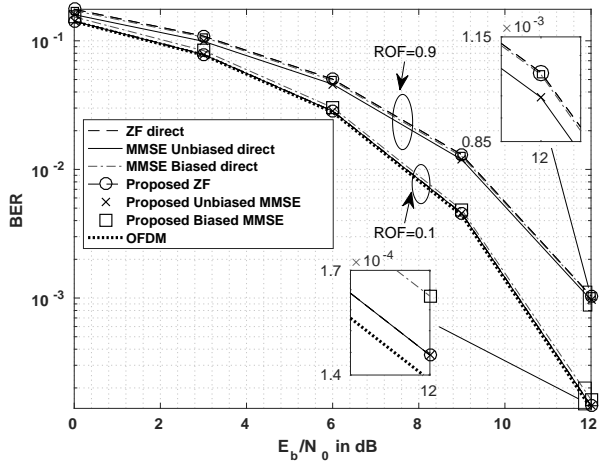}
\caption{ Uncoded BER performance of proposed GFDM transceiver and direct implementation GFDM transceiver for 16 QAM modulation in AWGN channel. Case I : $N=128$, $M=8$.}
\label{fig:ber128x8awgn}
\end{figure}
 BER performance of low complexity GFDM system over AWGN channel for Case I is provided in Fig~\ref{fig:ber128x8awgn}.  It is observed that there is no degradation in performance as compared to the direct implementation. For ROF value of 0.1, ZF and unbiased MMSE have similar performance whereas biased MMSE has worse performance than ZF and unbiased MMSE receivers. Bias correction provides SNR gain of 0.1 dB at the BER of $10^{-4}$. For ROF value of 0.9, unbiased MMSE receiver provides SNR gain of 0.1 dB at the BER of $10^{-3}$ over biased MMSE receiver which has similar performance to ZF receiver. It can be concluded that our proposed low complexity bias correction in MMSE holds importance.

\begin{figure}[h]
\includegraphics[width=\linewidth]{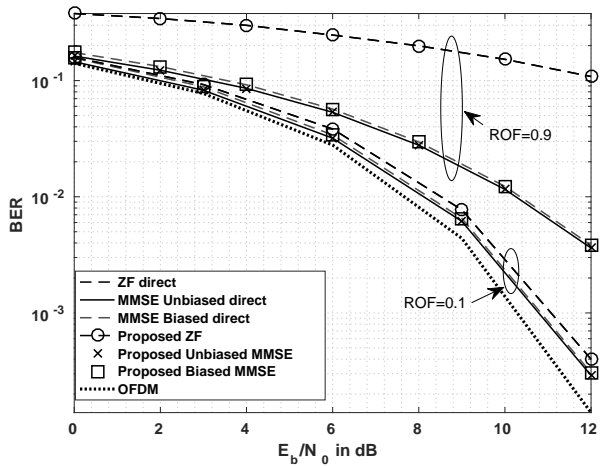}
\caption{Uncoded BER performance of proposed GFDM transceiver and direct implementation GFDM transceiver for 16 QAM modulation in AWGN channel. Case II : $N=8$, $M=128$.}
\label{fig:ber8x128awgn}
\end{figure}
\begin{figure}[h]
\includegraphics[width=\linewidth]{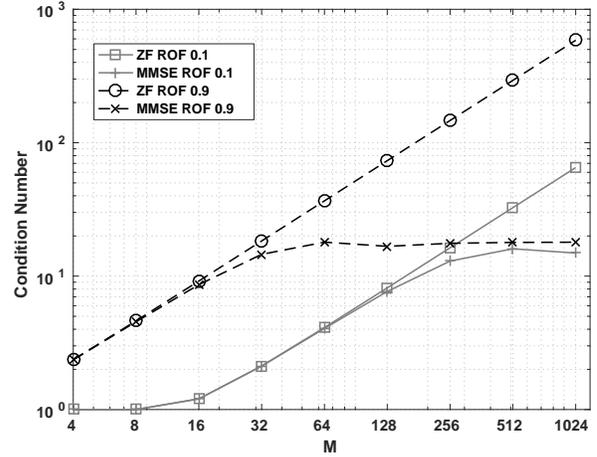}
\caption{Condition number of GFDM equalization matrix $\bA_{eq}$ for ZF and MMSE equalization for $M\in[2,~1024]$, $N=16$, ROF $\in\{0.1~0.9\}$ and $\snr=$30 dB.}
\label{fig:condition}
\end{figure}
 BER performance of low complexity GFDM system over AWGN channel for Case II is presented in Fig~\ref{fig:ber8x128awgn}.  It is observed that there is no degradation in performance as compared to the direct implementation. For ROF value of 0.1, our proposed unbiased MMSE receiver achieves SNR gain of 0.2 dB our proposed ZF receiver. For ROF value of 0.9, our proposed unbiased MMSE receiver outperforms our proposed ZF receiver.  To get more insight into behaviour of receivers at large $M$'s, condition number of $\bA_{eq}$ for ZF and MMSE receivers is plotted for $M\in[2,~1024]$, $N=16$, ROF $\in{0.1~0.9}$ and $\snr=$ 30 dB in Fig.~\ref{fig:condition}. Condition number of $\bA_{eq}$ for ZF receiver increases with $M$ which degrades it's performance as $M$ increases \cite{matthe_generalized_2014,nimr_optimal_2017}. However, condition number for $\bA_{eq}$ for MMSE receiver saturates at the $M=32$ and $M=256$ for ROF value of 0.1 and 0.9. This means that for large values of $M$'s, MMSE receiver  outperforms ZF receiver. Interestingly, performance gain achieved by our proposed MMSE receiver comes with mere 3 percent additional complexity over our proposed receiver.   

It can be concluded that our proposed ZF and MMSE low complexity self-interference equalizers do not incur any performance loss as they maintain their optimum performance. When $M$ is large, MMSE equalizer is preferred as it outperforms ZF equalizer as well as does not incur additional significant complexity over ZF equalizer.

\subsection{Multipath Fading Channel}
In this section, we present BER performance of our proposed low complexity transceiver in multipath fading channel.  We consider 3GPP extended typical urban channel (ETU) \cite{series2009guidelines} whose whose channel delay and channel power are [0 50 120 200 230 500 1600 2300 5000] $\mu$s and [-1 -1 -1 0 0 0 -3 -5 -7] dB, respectively. The
CP is chosen long enough to accommodate the wireless channel
delay spread. We consider a coded system to compare our results. We used convolutional code with code rate of 0.5 with a constraint length of 7 and code generator polynomials of 171 and 133.

For multipath fading channel Case I indicates a scenario where $\delta f_{gfdm} < \frac{B_c}{N}$ whereas Case II indicates a scenario where $\frac{B_c}{N}<\delta f_{gfdm} < \frac{B_c}{NM}$. For both cases, we consider $\delta f_{ofdm} < \frac{B_c}{N}$ .
\begin{figure}[h]
\includegraphics[width=\linewidth]{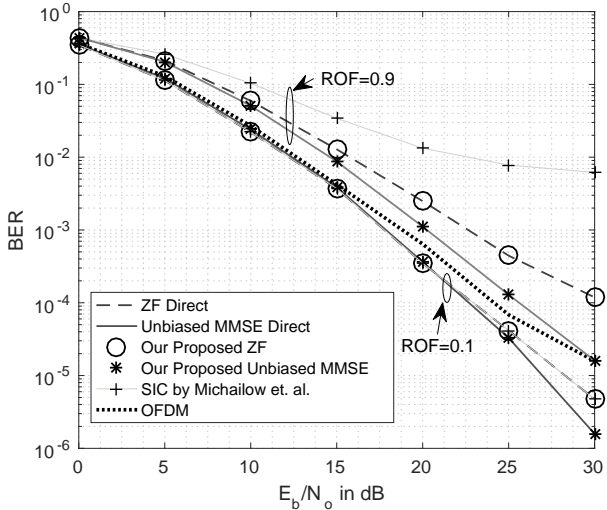}
\caption{Coded BER Performance of proposed GFDM transceiver and direct implementation GFDM transceiver for 16 QAM modulation in fading channel. Case I : $N=128$, $M=8$ and sub-carrier bandwidth is less than channel coherence bandwidth.}
\label{fig:ber128x8fading}
\end{figure}

BER of our proposed low complexity transceiver in multipath fading channel for Case I is plotted in Fig.~\ref{fig:ber128x8fading}. Our proposed receivers do not incur any performance loss over direct implementations. Our proposed MMSE receiver shows significant BER performance gain over other GFDM receivers. For ROF value of 0.1, our proposed MMSE receiver gives the best performance. MMSE receiver achieves SNR gain of  4 dB over OFDM at BER of $10^{-5}$. This SNR gain over OFDM is due to higher frequency resolution of GFDM \cite{michailow_generalized_2014}. MMSE receiver also achieves SNR gain of 2.5 dB over our proposed ZF receiver and SIC receiver in \cite{gaspar_low_2013} at BER of $4\times 10^{-6}$. For ROF value of 0.9, our proposed MMSE receiver shows a SNR gain of 5 dB over our proposed ZF receiver at BER of $10^{-4}$. BER of SIC receiver in \cite{gaspar_low_2013} floors at $10^{-2}$ and has SNR loss of 15 dB over our proposed MMSE receiver. 

\begin{figure}[h]
\includegraphics[width=\linewidth]{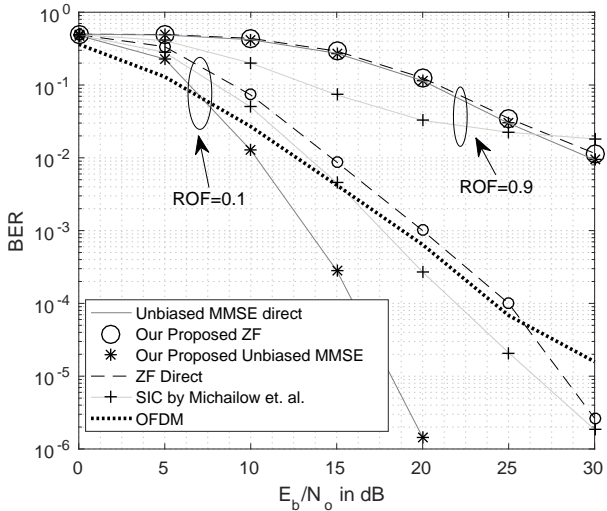}
\caption{ Coded BER performance of proposed GFDM transceiver and direct implementation GFDM transceiver for 16 QAM modulation in fading channel. Case II : $N=8$, $M=128$ and sub-carrier bandwidth is around ten times channel coherence bandwidth.}
\label{fig:ber8x128fading}
\end{figure}
 
BER of our proposed low complexity transceiver in multipath fading channel for Case II is plotted in Fig.~\ref{fig:ber8x128fading}. Our proposed receivers do not incur any performance loss over direct implementations. Since $\delta f_{gfdm}\sim 10\times B_{c}$, GFDM receivers show frequency diversity gain for ROF value of 0.1. Maximum diversity gain is extracted by our proposed MMSE receiver. Our proposed MMSE receiver has SNR gain of 12 dB, 9 dB and 8 dB over OFDM receiver, SIC receiver in \cite{gaspar_low_2013} and our proposed receiver respectively. When ROF value is 0.9, GFDM receivers show huge performance degradation.

 It can be concluded that our proposed MMSE receiver show significant performance gain over our proposed ZF receiver in \cite{gaspar_low_2013} in multipath fading channel with mere 3 percent additional computational load. Performance gain of our proposed MMSE receiver is achieved with 100 times lesser computational load.

\section{Conclusion}
In this work, we have proposed  low complexity GFDM transceivers. We used a special factorization of GFDM modulation matrix to design low complexity GFDM transceiver without incurring any performance loss. GFDM modulation matrix is factorized in terms of DFT matrices and diagonal matrix to design low complexity GFDM transmitter. This factorization was also used to derive closed form expression of MF, ZF and MMSE self-interference equalizers. We also derived closed form expression for bias correction of MMSE equalizer. These closed-form expressions lead to low complexity FFT based self-interference equalizers. Two stage receivers in which channel equalization is followed by self-interference equalizers was proposed for multipath fading channel. FFT based low complexity implementation of these receivers was presented.

  Computational complexity of our transceiver was computed and compared with the existing ones known so far to have the lowest complexity. Our proposed receivers MF, ZF and MMSE receivers are shown to have similar complexity and log-linear with number of symbols. We found that our transceivers show huge complexity reduction as compared  to ones in \cite{ michailow_generalized_2014, farhang_low-complexity_2016}. Our proposed transmitter and ZF receiver achieves around 100 times complexity reduction over ZF in \cite{farhang_low-complexity_2016}.  Over 300 times complexity reduction can be achieved through our MMSE receiver compared with the proposed MMSE receiver in \cite{farhang_low-complexity_2016}.  Our transceivers are also shown to have 2 to 9 times more complex than OFDM transceiver. When $M$ is large, our proposed MMSE receiver outperforms our proposed ZF receiver without any signification computational complexity addition. Such significant complexity reduction makes our transceiver an attractive choice for hardware implementation of GFDM systems. 

\begin{appendices}

\section{Proof of Lemma~\ref{th:factorization G}} \label{app:lem1}

The matrix $\bR$ can be written as,
\be
\bR = (\dftmat_M \kron \bI_N) \bS (\dftmat_M \kron \bI_N)^{\rm H}.
\ee
Using the properties of Kronecker product\cite{schacke_kronecker_2013}, $\bR$ can be further simplified to,
\begin{equation}
\begin{aligned}
\bR &= \permut^{\rm T} (\bI_N \kron \dftmat_M) \permut \bS \permut^{\rm T} (\bI_N \kron \dftmat_M^{\rm H}) \permut\\
&= \permut^{\rm T} \diagdft{M} \bar{\bS} \diagdft{M}^{\rm H} \permut,
\end{aligned} 
\end{equation}
where, $\bar{\bS}= \permut \bS \permut^{\rm T}$. Definition of $\bar{\bS}$, given in (\ref{eqn:sr}), can be directly obtained by using the definition of $\permut$ given in (\ref{eqn:permut}).

\section{Proof of Theorem~\ref{Th:modmat}} \label{app:Th1}
Using Lemma~\ref{th:factorization G} for $\bA$ given in (\ref{eqn:Modmatrix}), $\bA$ can be given as,
\be
\bA= \permut^{\rm T} \diagdft{M} \bar{\bD} \diagdft{M}^{\rm H} \permut \diagdft{N},  
\ee
where,  $\bar{\bD} = \permut \bD \permut^{\rm T}=diag\{\bar{\lambda}(0),~\bar{\lambda}(1)\cdots\bar{\lambda}(MN-1)\}$. Using (\ref{eqn:sr}) in (\ref{eqn:DiagonalvaluesofDmatrix}), the $r^{\rm th}$ diagonal value $\bar{\lambda}(r)$ can be given as,
\begin{equation} \label{eqn:proof:Th1}
\begin{aligned}
\bar{\lambda}(r)=\lambda((r \mod M)N+\floor*{\frac{r}{M}}), ~ \text{for} ~ 0\leq r \leq MN-1 \\
\begin{split}
&= \sum_{m=0}^{M-1}{g[mN+ ((r \mod M)N+\floor*{\frac{r}{M}}) \mod N]} \times \\ &{\omega^{-m\floor*{\frac{(r \mod M)N+\floor*{\frac{r}{M}})}{N}}}]}
\end{split}
\end{aligned}
\end{equation}
Now, using the fact that $\floor*{\frac{r}{M}}$ will vary from $0$ to $N-1$, $r \mod M $ will vary from $0$ to $M-1$, $((r \mod M)N+\floor*{\frac{r}{M}}) \mod N] = \floor*{\frac{r}{M}}$ and $\floor*{\frac{(r \mod M)N+\floor*{\frac{r}{M}})}{N}}= r \mod M$. (\ref{eqn:transmitclosedformdbar}) can be obtained by putting these simplified values in  (\ref{eqn:proof:Th1}).

\section{Proof of Theorem~\ref{Th:Aeq}} \label{app:proofTh2}
 Using the properties of unitary matrices, Theorem~\ref{Th:modmat} and (\ref{eqn:gfdmeualizermatrix}), $\bA_{eq}$ for MF and $\bA_{eq}^{\rm MF}$, ZF can be given as,
 \be \label{eqn:MFmatrix}
 \begin{aligned}
 \bA_{eq}^{\rm MF} &= [\permut^{\rm T} \diagdft{M} \bar{\bD} \diagdft{M}^{\rm H} \permut \diagdft{N}]^{\rm H}\\
 &= \diagdft{N}^{\rm H} \permut^{\rm T} \diagdft{M} \bar{\bD}^{\rm H} \diagdft{M}^{\rm H} \permut. 
 \end{aligned}
 \ee
 
 \be \label{eqn:ZFmatrix}
 \begin{aligned}
 \bA_{eq}^{\rm ZF} &= [\permut^{\rm T} \diagdft{M} \bar{\bD} \diagdft{M}^{\rm H} \permut \diagdft{N}]^{-1}\\
 &= \diagdft{N}^{\rm H} \permut^{\rm T} \diagdft{M} \bar{\bD}^{-1} \diagdft{M}^{\rm H} \permut.
 \end{aligned}
 \ee
 Now, let, $\unitarymatrix_1=\permut^{\rm T} \diagdft{M}$ and $\unitarymatrix_2= \diagdft{M}^{\rm H} \permut \diagdft{N}$. Both $\unitarymatrix_1$ and $\unitarymatrix_2$ are unitary matrix and $\bA=\unitarymatrix_1 \bar{\bD} \unitarymatrix_2$. Using this, $\AHA= \unitarymatrix_2^{\rm H} abs\{\bar{\bD}\}^2 \unitarymatrix_2$. Using above definitions and (\ref{eqn:gfdmeualizermatrix}), $\bA_{eq}^{MMSE}$ can be given as,
 \be \label{eqn:MMSEmatrix}
 \begin{aligned}
\bA_{eq}^{MMSE} &= [\mI+\unitarymatrix_2^{\rm H} abs\{\bar{\bD}\}^2 \unitarymatrix_2]^{-1} \unitarymatrix_2^{\rm H} \bar{\bD}^{\rm H} \unitarymatrix_1^{\rm H} \\
&= \unitarymatrix_2^{\rm H} [\mI+abs\{\bar{\bD}\}^2]^{-1} \bar{\bD}^{\rm H} \unitarymatrix_1^{\rm H} \\
&= \diagdft{N}^{\rm H} \permut^{\rm T} \diagdft{M}  [\mI+abs\{\bar{\bD}\}^2]^{-1} \bar{\bD}^{\rm H} \diagdft{M}^{\rm H} \permut.
\end{aligned}
 \ee
  Unbiased MMSE equalizer matrix, $\bA_{eq}^{un-MMSE}= \matrixmmsebias{gfdm} \bA_{eq}^{MMSE}$.  Now, using the definition of $\bA$, given in (\ref{eqn:Modmatrix}), $\AHA=\diagdft{N}^{\rm H} \bdft abs\{\bD\}^2 \bdft^{\rm H} \diagdft{N}=(\bI_M \kron \dftmat_N^{\rm H})(\dftmat_M \kron \bI_N)  abs\{\bD\}^2 (\dftmat_M^{\rm H} \kron \bI_N)(\bI_M \kron \dftmat_N)$. Using the properties of Kronecker product, this can be further simplified as, $\AHA=(\dftmat_M \kron \dftmat_N^{\rm H}) abs\{\bD\}^2 (\dftmat_M^{\rm H} \kron \dftmat_N)=\unitarymatrix_3 abs\{\bD\}^2 \unitarymatrix_3^{\rm H}$, where, $\unitarymatrix_3= \dftmat_M \kron \dftmat_N^{\rm H}$ is a unitary matrix. The matrix $\unitarymatrix_3$ can be written as,
  \be \label{eqn:unitarymarix3}
  \unitarymatrix_3=\begin{pmatrix}
  \dftmat_N^{\rm H} &  \dftmat_N^{\rm H} & \ldots &  \dftmat_N^{\rm H}\\
   \dftmat_N^{\rm H} &  \omega \dftmat_N^{\rm H} & \ldots &  \omega^{\rm M-1}\dftmat_N^{\rm H}\\
   \vdots & \vdots &\vdots & \vdots \\
    \dftmat_N^{\rm H} &  \omega^{\rm M-1}\dftmat_N^{\rm H} & \ldots &  \omega^{\rm (M-1)^2}\dftmat_N^{\rm H}
  \end{pmatrix}.
  \ee
   Using above given definition of $\AHA$, $\matrixmmsebias{gfdm}$ in (\ref{eqn:gfdmmatrixmmsebias}) can be given as,
   \be \label{eqn:gfdmmatrixmmsebias1}
   \begin{aligned}
   \matrixmmsebias{gfdm} &= diag\{\unitarymatrix_3 [\mI+abs\{\bD\}^2]^{-1} abs\{\bD\}^2 \unitarymatrix_3^{\rm H}\}\\
   &= diag\{\unitarymatrix_3 \tilde{\bD} \unitarymatrix_3^{\rm H}\},
   \end{aligned}
   \ee
  where, $\tilde{\bD}=[\mI+abs\{\bD\}^2]^{-1} abs\{\bD\}^2$ is a diagonal matrix. Using the definition of $\bD$ in (\ref{eqn:DiagonalvaluesofDmatrix}), $r^{\rm th}$ diagonal value of $\tilde{\bD}$ can be given as,
  \be \label{eqn:Dtilde}
  \tilde{\bD}(r,r) = \frac{|{\lambda}_r|^2}{|{\lambda}_r|^2+\snri},~ 0\leq r \leq MN-1.  
  \ee 
   It is straight forward to obtain (\ref{eqn:th:gfdmmmsebiasmatrix}) by using (\ref{eqn:unitarymarix3}-\ref{eqn:Dtilde}) and properties of DFT matrix.
\section{Comparison of flops for FFT/IFFT algorithm when $N$ or $M$ is small} \label{sec:app:fft}

\begin{figure}[h]
\includegraphics[width=\linewidth]{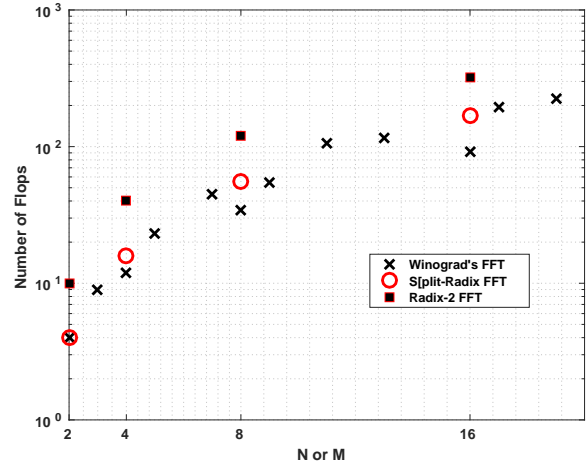}
\caption{Flops required for different FFT algorithms. This figure shows that Winograd's FFT achieves least complexity when $N$ or $M$ is small. }
\end{figure}
\end{appendices}
\bibliographystyle{IEEEtran}
\bibliography{Complexity,precoded_GFDM,frameAndgaborTh,FilteringPulseShaping,5G,books,GFDM_new,Comparision of Waveforms}
\end{document}